\documentclass[letterpaper,11pt]{article}

\usepackage[runin]{abstract}
\usepackage{geometry}
\usepackage{titlesec}
\usepackage{graphicx}
\usepackage{amsmath}
\usepackage{amsthm}
\usepackage{amsfonts}
\usepackage{amssymb}
\usepackage{empheq}
\usepackage{algorithmic}
\usepackage{algorithm}
\usepackage[authoryear]{natbib}
\usepackage{url}

\usepackage[font=sf]{caption}

\titleformat*{\section}{\Large\bfseries}
\titleformat*{\subsection}{\large\sc}
\titleformat*{\subsubsection}{\itshape}

\usepackage{mathtools}

\usepackage{epigraph}

\begin{document}

\title{{\bf A dynamic that evolves toward a Nash equilibrium}}

\author{\large{Ioannis Avramopoulos}\thanks{RelationalAI, Inc.}}

\date{}

\maketitle

\thispagestyle{empty} 

\newtheorem{definition}{Definition}
\newtheorem{proposition}{Proposition}
\newtheorem{theorem}{Theorem}
\newtheorem*{theorem*}{Theorem}
\newtheorem{corollary}{Corollary}
\newtheorem{lemma}{Lemma}
\newtheorem{axiom}{Axiom}
\newtheorem{thesis}{Thesis}

\vspace*{-0.2truecm}

\begin{abstract}
In this paper, we study an exponentiated multiplicative weights dynamic based on Hedge, a well-known algorithm in theoretical machine learning and algorithmic game theory. The empirical average (arithmetic mean) of the iterates Hedge generates is known to approach a minimax equilibrium in zero-sum games. We generalize that result to show that a weighted version of the empirical average converges to an equilibrium in the class of symmetric bimatrix games for a diminishing learning rate parameter. Our dynamic is the first dynamical system (whether continuous or discrete) shown to evolve toward a Nash equilibrium without assuming monotonicity of the payoff structure or that a potential function exists. Although our setting is somewhat restricted, it is also general as the class of symmetric bimatrix games captures the entire computational complexity of the PPAD class (even to approximate an equilibrium).
\end{abstract}

\section{Introduction}


{\em Game theory} is a mathematical discipline concerned with the study of algebraic, analytic, and other objects that abstract the physical world, especially social interactions. The most important {\em solution concept} in game theory is the {\em Nash equilibrium} \citep{Nash}, a strategy profile (combination of strategies) in an $N$-player game such that no unilateral player deviations are profitable. The Nash equilibrium is an attractive solution concept, for example, as Nash showed, an equilibrium is guaranteed to exist in any $N$-person game. Over time this concept has formed a basic cornerstone of {\em economic theory,} but its reach extends beyond economics to the natural sciences and biology.

One of the limitations of Nash equilibrium as a plausible solution concept is that we do not have an efficient algorithm for computing one. In fact, it has been conjectured that Nash equilibrium computation is intractable as it is complete for the complexity class PPAD \citep{Daskalakis, CDT}. This class, introduced by \cite{PPAD}, contains a variety of related problems (such as computing Brouwer fixed points) that we don't have efficient algorithms for. Thus, there is a gap between game theory and the theory of computing. In this paper, we take a step toward reconciling these disciplines using {\em dynamical systems theory} as an intermediate step. 

The notion of {\em equilibrium} admits various definitions in the mathematical sciences. One such standard definition is as a {\em fixed point} of a dynamical system. Research in dynamical systems is hardly content with identifying the fixed points of a (continuous) flow or a (discrete) map. What is ultimately important in this mathematical branch is to understand the dynamic evolution of system trajectories whether near fixed points (or near, for example, {\em limit cycles}) or globally. The Nash equilibrium can also be understood as a fixed point, for example, of the {\em best response correspondence} of a game (and other dynamics). But dynamical systems (as, for example, studied in theory of learning in games \citep{Fudenberg-Levine, Cesa-Bianchi} or evolutionary game theory \citep{Weibull, PopulationGames}) whose trajectories evolve toward a Nash equilibrium (generically, without restrictive assumptions on the payoff structure) elude us. Since dynamical systems are algorithms, progress in this direction evidently informs algorithmic research.

\subsection{Our question}

There are classes of games where equilibrium computation is known to be tractable: {\em Zero-sum games} are equivalent to linear programming and thus minimax equilibrium computation admits a polynomial-time algorithm. A bimatrix game is called {\em symmetric} if the payoff matrix of each player is the transpose of that of other. Every symmetric $N$-person game (and, thus every symmetric bimatrix game) admits a symmetric Nash equilibrium \citep{Nash2}.  If the payoff matrix of the symmetric game is also symmetric, then the game is called {\em doubly symmetric}. Symmetric equilibrium computation in doubly symmetric games admits a fully polynomial time approximation scheme (FPTAS) \citep{Ye}. In this paper, we consider the problem of computing (using the term {\em computation} in a broad sense as a dynamical system is not necessarily a Turing machine) a symmetric equilibrium in a symmetric bimatrix game (that is not necessarily doubly symmetric).

\cite{CDT} show that finding a Nash equilibrium in a $2$-person game is a PPAD-complete problem and that an equilibrium FPTAS in these games (under either of ``additive" or ``multiplicative" notions of payoff approximation) implies P = PPAD. \cite{Avramopoulos2}, drawing on \citep{Jurg}, shows that an FPTAS for a symmetric equilibrium in symmetric bimatrix games also implies that P=PPAD. Thus approximating an equilibrium in the class of games we consider is conjectured to be a hard problem. In this paper, we question that belief. Our result was motivated by an elementary question concerning {\em multiplicative weights dynamics}.

An important result at the intersection of theoretical machine learning and game theory is that the {\em empirical average} of the strategies generated by a well-known multiplicative weights algorithm, namely, Hedge \citep{FreundSchapire1, FreundSchapire2}, approaches the minimax strategy of a respective zero-sum game. In doubly symmetric bimatrix games, Hedge is a multiplicative version of {\em gradient ascent} (with an additional gradient exponentiation step). Gradient ascent is known to converge to critical points of nonlinear optimization problems (for example, see \citep[p. 48]{Bertsekas}) and symmetric equilibrium computation in a doubly symmetric game is a special case of a quadratic programming. It is natural to expect that by tuning Hedge's learning-rate parameter, obtaining convergence in doubly symmetric games falls within the realm of possibilities of this algorithm.

But every symmetric bimatrix game is the sum of a doubly symmetric game and a symmetric zero-sum game as every matrix can be decomposed as (see \citep{Horn}): 
\begin{align*}
C = \frac{1}{2} (C+C^T) + \frac{1}{2} (C-C^T).
\end{align*}
Since the iterates of Hedge are expected to converge to a symmetric equilibrium in a doubly symmetric game and the empirical average of the iterates is known to converge to such an equilibrium in symmetric zero-sum games,
is it then possible that the empirical average of iterates converges to a symmetric equilibrium in the general class of symmetric bimatrix games? Our answer is affirmative.

\subsection{Our results and techniques}

Our analysis rests on using relative entropy and the related logarithmic function as potentials. That relative entropy facilitates analyzing multiplicative weights (and its continuous approximation, namely, the {\em replicator dynamic}) is well known (for example, see \citep{Weibull, FreundSchapire2}). Our use of the logarithmic function as a potential is related to a proof that the time average of the trajectory of the replicator dynamic in a symmetric bimatrix game with an interior equilibrium converges to that equilibrium \citep{Wolff} (see also \citep[p. 91]{Weibull}). Our main tool is that the composition of the relative entropy function with Hedge is a convex function of its learning rate parameter a result appears for the first time in \citep{Avramopoulos1}.\footnote{In that manuscript, I erroneously believed to have shown P = PPAD. The error is in Lemma 10.} 

\subsubsection{Evolution under a diminishing learning rate}

Our main result is Theorem \ref{asymptotic_convergence_theorem} showing that, under standard scheduling conditions for the learning rate, every limit point of the sequence of {\em weighted} averages of the iterates is a symmetric equilibrium strategy. (To compute the empirical average, our algorithm weighs each iterate with the value of the learning rate in the respective iteration.) The {\em Bolzano-Weierstrass theorem} then yields as a corollary that a symmetric equilibrium always exists in symmetric bimatrix games. Since our equilibrium computation problem is PPAD-complete, we obtain alternative proofs of existence of equilibria (fixed points) in a variety of related problems (such as $N$-person games). 

The rest of our paper is a deeper exploration of the ideas and techniques used in Theorem \ref{asymptotic_convergence_theorem}. Natural questions to ask in this setting are: (1) Are the elements of the sequence of empirical averages approximate equilibria? (2) Can we be more precise on the structure of the limit set? 

Question (1) is answered in Theorem \ref{payoff_of_average_strategy_theorem} and Corollary \ref{approximate_equilibria_corollary} showing that for any approximation error, say $\epsilon$, every element far enough in the sequence of empirical averages is an $\epsilon$-approximate equilibrium. This is an intuitively pleasing result that the sequence of empirical averages is ``tame.'' 

Question (2) is answered in Theorem \ref{equilibrium_polytope_limitset_theorem} showing that the affine hull of the limit set is a connected subset of a convex polytope of equilibria. To prove that the limit set is connected, we show that the distance between successive elements of the sequence of empirical averages converges to zero (Lemma \ref{consecutive_empirical_averages_lemma}) and then invoke standard results from topology (Theorem \ref{limitpoint_connectedness_theorem}). To prove that a convex polytope of equilibria contains the limit set we rely on a structural property of equilibria (Lemma \ref{equilibrium_set_affine_hull_lemma}) that to the extent of the our knowledge does not appear in that form in the literature.

\subsubsection{Equilibrium computation using side information}
 
The results we have insofar discussed derive an (approximate) equilibrium by inspecting the empirical average (or an iterate). Our final thread of results (cf. Section \ref{side_information}) explores how ``side information'' can facilitate equilibrium computation. For example, if the sequence of empirical averages converges, then the probability masses of those pure strategies that support the limit eventually dominate (are greater than) the probability masses of strategies outside that support giving rise to a ``separation phenomenon'' that yields a finite algorithm to compute an (exact) equilibrium through checking whether the ``heaviest'' strategies support an equilibrium.\footnote{An operation that can be performed in polynomial time using linear programming \citep{CDT}} There are a variety of ways to take these observations further: We explore ranking pure strategies based on average probability masses (which can be effective even if the limit set is nontrivial) (Theorem\ref{recurrent_average_probability_masses}), based on the payoffs pure strategies yield against the empirical average (Theorem \ref{equalizer_separation_theorem}), and based on the probability masses of the iterates (Corollary \ref{probability_mass_ranking_corollary}). Our goal is to reveal structure in the sequences of iterates and empirical averages toward basic computational goals (but our results that may also bear relevance to natural and social phenomena where game theory applies). In this vein, we also explore phenomena of recurrence and ``clean separation'' using best responses (e.g., Theorem \ref{persistent_best_responses}).

\subsection{Other related work regarding equilibrium computation algorithms}

There is a significant amount of work on Nash equilibrium computation especially in the setting of $2$-player games. We simply mention a boundary of those results. The Lemke-Howson algorithm for computing an equilibrium in a bimatrix game is considered by many to be the state of the art in exact equilibrium computation but it has been shown to run in exponential time in the worst case \citep{Savani-Stengel}. There is a quasi-polynomial algorithm for additively approximate Nash equilibria in bimatrix games due to \cite{LMM} (based on looking for equilibria of small support on a grid). The best polynomial-time approximation algorithm for a Nash equilibrium achieves a 0.3393 approximation \citep{Tsaknakis-Spirakis-journal}.

\subsection{Roadmap for the rest of the paper}

We start with preliminary background on game theory in Section \ref{preliminaries}. The convexity lemma of multiplicative weights and various implications is discussed in Section \ref{convexity_lemma_section}. Section \ref{asymptotic_convergence_section} contains our main result that, under a diminishing learning rate, the limit points of the sequence of empirical averages are symmetric equilibrium strategies. In Section \ref{equilibrium_approximation_VLR}, we consider the problem of equilibrium approximation under a finite time horizon (under the same assumption that the learning rate obeys a diminishing schedule). In Section \ref{structure_of_limit_set_section}, we shed light on the structure of the limit set of the sequence of weighted averages proving it is a connected set inside a convex equilibrium polytope. Finally, in Section \ref{side_information}, we study how side information can inform the equilibrium computation process.

\section{Equilibrium computation background and preliminary results}
\label{preliminaries}

\subsection{Symmetric bimatrix games}

A $2$-player (bimatrix) game in normal form is specified by a pair of $n \times m$ matrices $A$ and $B$, the former corresponding to the {\em row player} and the latter to the {\em column player}. A {\em mixed strategy} for the row player is a probability vector $P \in \mathbb{R}^n$ and a mixed strategy for the column player is a probability vector $Q \in \mathbb{R}^m$. The {\em payoff} to the row player of $P$ against $Q$ is $P \cdot A Q$ and that to the column player is $P \cdot B Q$. Let us denote the space of probability vectors for the row player by $\mathbb{P}$ and the corresponding space for the column player by $\mathbb{Q}$. A Nash equilibrium of a $2$-player game $(A, B)$ is a pair of mixed strategies $P^*$ and $Q^*$ such that all unilateral deviations from these strategies are not profitable, that is, for all $P \in \mathbb{P}$ and $Q \in \mathbb{Q}$, we simultaneously have that
\begin{align}
P^* \cdot AQ^* &\geq P \cdot AQ^*\label{eqone}\\
P^* \cdot BQ^* &\geq P^* \cdot BQ.\label{eqtwo}
\end{align}
We denote the set of Nash equilibria of $(A, B)$ by $NE(A, B)$. If $B = A^T$, where $A^T$ is the transpose, $(A, B)$ is called {\em symmetric}. Let $(C, C^T)$ be a symmetric bimatrix game. We call $(P^*, Q^*) \in NE(C, C^T)$  symmetric if $P^* = Q^*$. If $(X^*, X^*)$ is a symmetric equilibrium, we call $X^*$ a symmetric equilibrium strategy. We denote the set of symmetric equilibrium strategies of $(C, C^T)$ by $NE^+(C)$.

\subsubsection{Some further notation}

We denote the space of symmetric bimatrix games by $\mathbb{C}$. $\mathbb{C}^+$ is the subset of $\mathbb{C}$ such that the payoff matrix is nonnegative. If the payoff entries lie in the range $[0, 1]$, we denote the corresponding space by $\mathbb{\hat{C}}$. Given $C \in \mathbb{C}$, we denote the corresponding set of pure strategies by $\mathcal{K}(C) = \{1, \ldots, n\}$. Pure strategies are denoted either as $i$ or as $E_i$, a probability vector whose mass is concentrated in position $i$. $\mathbb{X}(C)$ is the probability simplex (space of mixed strategies) corresponding to $C \in \mathbb{C}$. We denote the (relative) interior of $\mathbb{X}(C)$ by $\mathbb{\mathring{X}}(C)$ (every pure strategy in $\mathbb{\mathring{X}}(C)$ has probability mass). Let $X \in \mathbb{X}(C)$. We define the {\em support} or {\em carrier} of $X$ by
\begin{align*}
\mathcal{C}(X) \equiv \{ i \in \mathcal{K}(C) | X(i) > 0\}.
\end{align*}

\subsubsection{Approximate equilibria and best responses}

Conditions \eqref{eqone} and \eqref{eqtwo} simplify as follows for a symmetric equilibrium strategy $X^*$:
\begin{align*}
\forall X \in \mathbb{X}(C) : (X^* - X) \cdot CX^* \geq 0.
\end{align*}
An $\epsilon$-approximate symmetric equilibrium satisfies:
\begin{align*}
\forall X \in \mathbb{X}(C) : (X^* - X) \cdot CX^* \geq -\epsilon.
\end{align*}
We may equivalently write the previous expression as
\begin{align*}
(CX^*)_{\max} - X^* \cdot CX^* \leq \epsilon,
\end{align*}
where
\begin{align*}
(CX^*)_{\max} = \max\{ Y \cdot CX^* | Y \in \mathbb{X}(C) \}.
\end{align*}
If $Y^+ \cdot CX^* = (CX^*)_{\max}$, then $Y^+$ is called a best response to $X^*$. We note in passing the standard property that if $Y^+$ is a best response to $X^*$, for all $i \in \mathcal{C}(Y^+)$, $E_i$ is also a best response to $Y^+$.

\begin{definition}
$Y^+$ is a $\epsilon$-approximate best response to $X$, if $(CX)_{\max} - Y^+ \cdot CX \leq \epsilon$.
\end{definition}

\begin{definition}
\label{cool}
$X^*$ is an $\epsilon$-well-supported symmetric equilibrium strategy of $C \in \mathbb{C}$ if
\begin{align*}
(CX^*)_i > (CX^*)_j + \epsilon \Rightarrow X^*(j) = 0.
\end{align*}
\end{definition}

Well-supported equilibria admit the following characterization:

\begin{lemma}
$X^*$ is an $\epsilon$-well-supported Nash equilibrium of $C \in \mathbb{C}$ if and only if
\begin{align*}
X^*(i) > 0 &\Rightarrow (CX^*)_i \geq (CX^*)_{\max} - \epsilon.
\end{align*}
\end{lemma}

\begin{proof}
The statement of the lemma is just the contrapositive of Definition \ref{cool}.
\end{proof}

We further introduce the following definition:

\begin{definition}
$Y^+$ is a {\em well-supported} $\epsilon$-approximate best response to $X$, if, for all $i \in \mathcal{C}(Y^+)$, the pure strategies $E_i$ are also $\epsilon$-approximate best responses to $X$.
\end{definition}

\subsection{Equalizers and equilibrium subequalizers}

\begin{definition}
\label{equalizer_definition}
$X^* \in \mathbb{X}(C)$ is called an {\em equalizer} if
\begin{align*}
\forall X \in \mathbb{X}(C) : (X^* - X) \cdot CX^* = 0. 
\end{align*} 
We denote the set of equalizers of $(C, C^T)$ by $\mathbb{E}(C)$.
\end{definition}

Note that $\mathbb{E}(C) \subseteq NE^+(C)$. Equalizers generalize interior symmetric equilibrium strategies, as every such strategy is an equalizer, but there exist symmetric bimatrix games with a non-interior equalizer (for example, if a column of $C$ is constant, the corresponding pure strategy of $(C, C^T)$ is an equalizer of $(C, C^T)$). Note that an equalizer can be computed in polynomial time by solving the linear (feasibility) program (LP)
\begin{align*}
(CX)_1 = \cdots = (CX)_n, \quad \sum_{i=1}^n X(i) = 1, \quad X \geq 0,
\end{align*}
which we may equivalently write as
\begin{align*}
CX = c \mathbf{1}, \quad \mathbf{1}^T X = 1, \quad X \geq 0,
\end{align*}
where $\mathbf{1}$ is a column vector of ones of appropriate dimension. We may write this problem as a standard LP as follows: Letting
\begin{align*}
A \doteq \left[ \begin{array}{cc}
C & - \mathbf{1} \\
\mathbf{1}^T & 0 \end{array} \right]
\mbox{ and }
Y \doteq \left[ \begin{array}{c}
X \\
c \end{array} \right],
\end{align*}
we obtain
\begin{align*}
A Y = \left[ \begin{array}{cc}
C & - \mathbf{1} \\
\mathbf{1}^T & 0 \end{array} \right]
\left[ \begin{array}{c}
X \\
c \end{array} \right] =
\left[ \begin{array}{c}
CX - c \mathbf{1} \\
\mathbf{1}^T X \end{array} \right],
\end{align*}
and the standard form of our LP, assuming $C > 0$, is
\begin{align}
\left[ \begin{array}{c}
CX - c \mathbf{1} \\
\mathbf{1}^T X \end{array} \right] = \left[ \begin{array}{c}
\mathbf{0} \\
1 \end{array} \right], X \geq 0, c \geq 0\label{my_LP}
\end{align}
where $\mathbf{0}$ is a column vector of zeros of appropriate dimension.

\subsubsection{On the structure of equalizers}

From the previous discussion, we immediately obtain that:

\begin{lemma}
$\mathbb{E}(C)$ is a convex set.
\end{lemma}

\begin{proof}
The set of feasible/optimal solutions of a linear program is a convex set. Let
\begin{align*}
\mathbb{Y}^* = \left\{ [X^T \mbox{ } c]^T | [X^T \mbox{ } c]^T \mbox{ is a feasible solution to \eqref{my_LP}} \right\}.
\end{align*}
Then $\mathbb{Y}^*$ is convex and therefore the set
\begin{align*}
\mathbb{X}^* = \left\{ X | [X^T \mbox{ } c]^T \mbox{ is a feasible solution to \eqref{my_LP}} \right\}
\end{align*}
is also convex.
\end{proof}

We can actually show something stronger:

\begin{lemma}
\label{lman}
If $X_1^*, X_2^* \in \mathbb{E}(C)$ then $\Lambda(X_1^*, X_2^*) \subseteq \mathbb{E}(C)$, where
\begin{align*}
\Lambda(X_1^*, X_2^*) = \left\{ (1-\lambda) X_1^* + \lambda X^*_2 \in \mathbb{X}(C) | \lambda \in \mathbb{R} \right\},
\end{align*}
that is, $\Lambda(X_1^*, X_2^*)$ is the restriction of the affine hull of $X_1^*$ and $X_2^*$ to the simplex $\mathbb{X}(C)$.
\end{lemma}

\begin{proof}
Assume $X_1^*, X_2^* \in \mathbb{E}(C)$. Then, by the definition of an equalizer,
\begin{align*}
\forall X \in \mathbb{X}(C) : X^*_1 \cdot CX^*_1 &= X \cdot CX^*_1 \mbox{ and }\\
\forall X \in \mathbb{X}(C) : X^*_2 \cdot CX^*_2 &= X \cdot CX^*_2.
\end{align*}
Let
\begin{align*}
Y^* = (1-\lambda) X_1^* + \lambda X^*_2, \mbox{ } \lambda \in \mathbb{R}.
\end{align*}
Then
\begin{align*}
Y^* \cdot CY^* &= (1-\lambda) \left( (1-\lambda) X^*_1 + \lambda X^*_2 \right) \cdot CX^*_1 + \lambda \left( (1-\lambda) X^*_1 + \lambda X^*_2 \right) \cdot CX^*_2\\
  &= (1-\lambda) X^*_1 \cdot CX^*_1 + \lambda X^*_2 \cdot CX^*_2\\
  &= (1-\lambda) X \cdot CX^*_1 + \lambda X \cdot CX^*_2\\
  &= X \cdot C \left((1-\lambda) X^*_1 + \lambda X^*_2 \right)\\
  &= X \cdot CY^*.
\end{align*}
Since $X$ is arbitrary, the proof is complete.
\end{proof}

Equalizers admit a characterization as optimizers of a convex optimization problem:

\begin{lemma}
\label{equalizer_lemma}
$X^* \in \mathbb{E}(C)$ if and only if $\min\{(CX)_{\max} - (CX)_{\min}\} = 0$.
\end{lemma}

\begin{proof}
Note that for all $X, Y \in \mathbb{X}(C)$, $(CX)_{\max} \geq Y \cdot CX \geq (CX)_{\min}$. Assume now $X^*$ is an equalizer. Then it follows by Definition \ref{equalizer_definition} that $X^* \cdot CX^* = (CX^*)_{\max} - (CX^*)_{\min}$, implying $\min\{(CX)_{\max} - (CX)_{\min}\} = 0$. Conversely, if $\min\{(CX)_{\max} - (CX)_{\min}\} = 0$, then for the minimizer, say $X^*$, $(CX^*)_{\max} = X^* \cdot CX^* = (CX^*)_{\min}$, and the lemma follows.
\end{proof}

The optimization problem in the previous lemma can formulated as the linear program
\begin{align*}
\min \mbox{ } &\epsilon\\
\mbox{s.t. } &\forall i, j \in \mathcal{K}(C) \mbox{ such that } i \neq j: (CX)_i - (CX)_j \leq \epsilon\\
  &X \in \mathbb{X}(C).
\end{align*}

\subsubsection{Equilibrium subequalizers}
\label{christmastree}

\begin{definition}
Given a symmetric bimatrix game $C \in \mathbb{C}$ and a subset $\mathcal{L}$ of pure strategies of $C$, a {\em subequalizer} of $C$ given $\mathcal{L}$ is, if it exists, an equalizer of the subgame $C'$ whose carrier is $\mathcal{L}$. Note that a subequalizer may not be an equilibrium of $C$. An {\em equilibrium
subequalizer} of $C$ given $\mathcal{L}$, if it exists, is an equilibrium of $C$ that is an equalizer of $C'$. An {\em $\epsilon$-approximate equalizer} of $C'$ is a strategy $X^* \in \mathbb{X}(C')$ such that $(C'X^*)_{\max} - (C'X^*)_{\min} \leq \epsilon$. A {\em well-supported $\epsilon$-approximate equilibrium subequalizer}, say $X^*$, of $C$ given $\mathcal{L}$ is an $\epsilon$-approximate equilibrium of $C$ that satisfies 
\begin{align*}
(C'X^*)_{\max} - (C'X^*)_{\min} \leq \epsilon, \mbox{ } \forall i \in \mathcal{L} \mbox{ and } \forall j \not\in \mathcal{L}: (CX)_i \geq (CX)_j, \mbox{ and } X^*(j) = 0, \mbox{ } j \not\in \mathcal{L}.
\end{align*}
\end{definition}

We can check in polynomial time whether a subset of pure strategies support a well-supported $\epsilon$-approximate equilibrium subequalizer. Consider a subset $\mathcal{L}$ of pure strategies. We can find the best approximate equilibrium subequalizer these strategies support in polynomial time using the optimization problem
\begin{align*}
\min \mbox{ } &\epsilon\\
\mbox{s.t. } &\forall i \in \mathcal{L}: (CX)_{\max} - (CX)_i \leq \epsilon\\
  &\forall i \in \mathcal{L} \mbox{ and } \forall j \not\in \mathcal{L}: (CX)_i \geq (CX)_j\\
  &X \in \mathbb{X}(C) \mbox{ and } \forall j \not\in \mathcal{L}: X(j) = 0
\end{align*}
which is a convex problem that can be reformulated as the following linear program\footnote{I would like to thank Paul Spirakis for a related email communication that gave rise to this formulation.}
\begin{align*}
\min \mbox{ } &\epsilon\\
\mbox{s.t. } &\forall i, j \in \mathcal{L} \mbox{ such that } i \neq j: (CX)_i - (CX)_j \leq \epsilon\\
  &\forall i \in \mathcal{L} \mbox{ and } \forall j \not\in \mathcal{L}: (CX)_i \geq (CX)_j\\
  &X \in \mathbb{X}(C) \mbox{ and } \forall j \not\in \mathcal{L}: X(j) = 0.
\end{align*}

\subsection{On the structure of equilibrium sets}

The material we have insofar introduced are used (in Lemma \ref{equilibrium_set_affine_hull_lemma}) to obtain a property of equilibrium sets that is important in the sequel (as we study the limit points of our dynamical system).

\begin{lemma}
\label{best_responses_convexity_lemma}
Let $C \in \mathbb{C}$ and $i \in \mathcal{K}(C)$. The set $\left\{ X \in \mathbb{X}(C) | (CX)_i = (CX)_{\max} \right\}$ is convex.
\end{lemma}

\begin{proof}
The function
\begin{align*}
F(X) = (CX)_{\max} - (CX)_i
\end{align*}
is convex. Therefore, the set of minimizers of this function is also convex.
\end{proof}

\begin{lemma}
\label{equilibrium_set_affine_hull_lemma}
Let $X^*_1, X^*_2 \in NE^+(C)$ be best responses to each other and let $A(X^*_1, X^*_2)$ denote their affine hull (restricted to $\mathbb{X}(C)$). Then $X^* \in A(X^*_1, X^*_2)$ implies $X^* \in NE^+(C)$.
\end{lemma}

\begin{proof}
Let $\mathcal{K}^*_1 = \mathcal{C}(X^*_1)$ and $\mathcal{K}^*_2 = \mathcal{C}(X^*_2)$. Let $C_{12}$ be a subgame of $C$ such that $\mathcal{K}(C_{12}) = \mathcal{K}^*_1 \cup \mathcal{K}^*_2$. Then $A(X^*_1, X^*_2) \subset \mathbb{X}(C_{12})$. Since $X^*_2$ is a best response to $X^*_1$, for all $i \in \mathcal{K}^*_2$, $E_i$ is a best response to $X^*_1$. Therefore, $X^*_1$ is an equalizer of $C_{12}$ and by an analogous argument $X^*_2$ is also an equalizer of $C_{12}$. Thus, Lemma \ref{lman} implies every element of $A(X^*_1, X^*_2)$ is an equalizer of $C_{12}$. Let $X^*$ be an arbitrary element of $A(X^*_1, X^*_2)$. Then the assumptions that $X^*_1$ and $X^*_2$ are best responses to themselves (being equilibria) and to each other together with Lemma \ref{best_responses_convexity_lemma} imply that $X^*_1$ and $X^*_2$ are best responses to $X^*$ and thus every pure strategy of $C_{12}$ is a best response to $X^*$. Therefore, $X^* \in NE^+(C)$ and since $X^*$ is an arbitrary element of $A(X^*_1, X^*_2)$, the proof is complete.
\end{proof}

\section{The convexity lemma of exponentiated multiplicative weights}
\label{convexity_lemma_section}

Hedge \citep{FreundSchapire1, FreundSchapire2} induces the following map in our setting:
\begin{align}
T_i(X) = X(i) \cdot \frac{\exp\left\{ \alpha E_i \cdot CX \right\}}{ \sum_{j=1}^n X(j) \exp \left\{ \alpha E_j \cdot CX \right\} }, \quad i = 1, \ldots, n.\label{main_exp}
\end{align}
$\alpha$ is called the {\em learning rate}. Starting at an interior to the simplex strategy, say $X^0$, Hedge generates a {\em sequence of iterates} $\left\{ X^k \right\}_{k=0}^{\infty}$ such that $X^{k+1} = T(X^k)$ that remains in the interior for any finite number of iterations (but may reach the boundary in the asymptotic limit). We derive the {\em sequence of empirical averages} $\left\{ \bar{X}^K \right\}_{K = 0}^{\infty}$ as a weighted arithmetic mean
\begin{align*}
\bar{X}^K = \frac{1}{A_K} \sum_{k=0}^K \alpha_k X^k, \mbox{ where } A_K = \sum_{k = 0}^K \alpha_k
\end{align*}
and $\alpha_k > 0$ is the learning rate parameter used in step $k$.

\subsection{The multiplicative weights convexity lemma}

Let us now give some preliminary results on Hedge dynamics. Part of our analysis of Hedge relies on the relative entropy function between probability distributions (also called {\em Kullback-Leibler divergence}). The relative entropy between the $n \times 1$ probability vectors $P > 0$ (that is, for all $i = 1, \ldots, n$, $P(i) > 0$) and $Q > 0$ is given by 
\begin{align*}
RE(P, Q) \doteq \sum_{i=1}^n P(i) \ln \frac{P(i)}{Q(i)}.
\end{align*}
However, this definition can be relaxed: The relative entropy between $n \times 1$ probability vectors $P$ and $Q$ such that, given $P$, for all $Q \in \{ \mathcal{Q} \in \mathbb{X} | \mathcal{C}(P) \subset \mathcal{C}(\mathcal{Q}) \}$, is
\begin{align*}
RE(P, Q) \doteq \sum_{i \in \mathcal{C}(P)} P(i) \ln \frac{P(i)}{Q(i)}.
\end{align*}
We note the well-known properties of the relative entropy \cite[p.96]{Weibull} that {\em (i)} $RE(P, Q) \geq 0$, {\em (ii)} $RE(P, Q) \geq \| P - Q \|^2$, where $\| \cdot \|$ is the Euclidean distance, {\em (iii)} $RE(P, P) = 0$, and {\em (iv)} $RE(P, Q) = 0$ iff $P = Q$. Note {\em (i)} follows from {\em (ii)} and {\em (iv)} follows from {\em (ii)} and {\em (iii)}. With this background in mind, we state the following lemma (which we refer to as the {\em exponentiated multiplicative weights convexity lemma}) generalizing \cite[Lemma 2]{FreundSchapire2}.

\begin{lemma}[\citep{Avramopoulos2}]
\label{convexity_lemma}
Let $T$ be as in \eqref{main_exp}. Then
\begin{align*}
\forall X \in \mathbb{\mathring{X}}(C) \mbox{ } \forall Y \in \mathbb{X}(C) : RE(Y, T(X)) \mbox{ is a convex function of }\alpha.
\end{align*}
\end{lemma}

\subsection{Upper relative-entropy bounds}

A simpler version of the next lemma (assuming $C \in \mathbb{\hat{C}}$) is shown in \citep{Avramopoulos2}. Here we assume a general upper bound on the payoff entries. In the proof, we use the following ``secant inequality'' for a convex function $F(\cdot)$ and its derivative $F'(\cdot)$:
\begin{align}
\forall \mbox{ } b > a : F'(a) \leq \frac{F(b) - F(a)}{b - a} \leq F'(b).\label{secant_inequality}
\end{align}

\begin{lemma}
\label{convexity_lemma_not_normalized}
Let $C \in \mathbb{C}^+$. Then, for all $Y \in \mathbb{X}(C)$ and for all $X \in \mathbb{\mathring{X}}(C)$, we have that
\begin{align*}
\forall \alpha > 0 : RE(Y, T(X)) \leq RE(Y, X) - \alpha (Y-X) \cdot CX + \alpha (\exp\{m \alpha\} - 1) \bar{C},
\end{align*}
where $m = \max_{ij} C_{ij}$ and $\bar{C} > 0$ can be chosen independent of $X$ and $Y$.
\end{lemma}

\begin{proof}
Since, by Lemma \ref{convexity_lemma}, $RE(Y, T(X)) - RE(Y, X)$ is a convex function of $\alpha$, we have by the aforementioned secant inequality that, for $\alpha > 0$,
\begin{align}
RE(Y, T(X)) - RE(Y, X) \leq \alpha \left( RE(Y, T(X)) - RE(Y, X) \right)' = \alpha \cdot \frac{d}{d \alpha} RE(Y, T(X))\label{ooone}
\end{align}
where it can be readily computed that
\begin{align*}
\frac{d}{d \alpha} RE(Y, T(X)) = \frac{\sum_{j = 1}^n X(j) (CX)_j \exp\{ \alpha (CX)_j \}}{\sum_{j = 1}^n X(j) \exp\{ \alpha (CX)_j \}} - Y \cdot CX.
\end{align*}
Using Jensen's inequality in the previous expression, we obtain 
\begin{align}
\frac{d}{d \alpha} RE(Y, T(X)) \leq \frac{\sum_{j = 1}^n X(j) (CX)_j \exp\{ \alpha (CX)_j \}}{\exp\{ \alpha X \cdot CX \}} - Y \cdot CX.\label{vbvbvb}
\end{align}
Note now that
\begin{align*}
\exp\{ \alpha x \} \leq 1 + (\exp\{ \alpha \} - 1) x, x \in [0, 1],
\end{align*}
an inequality used in \citep[Lemma 2]{FreundSchapire2}. Letting $0 \leq y \leq m$, the previous inequality yields
\begin{align*}
\exp\{ \alpha y \} = \exp\left\{ m \alpha \frac{y}{m} \right\} \leq 1 + \frac{\exp\{ m \alpha \} - 1}{m} y.
\end{align*}

Using the latter inequality, we obtain from \eqref{vbvbvb} that
\begin{align*}
\frac{d}{d \alpha} RE(Y, T(X)) \leq \frac{X \cdot CX}{\exp\{ \alpha X \cdot CX \}} - Y \cdot CX + \frac{\exp\{ m \alpha \} - 1}{m} \frac{\sum_{j=1}^n X(j) (CX)_j^2}{\exp\{ \alpha X \cdot CX \}}
\end{align*}
and since $\exp\{\alpha X \cdot CX\} \geq 1$ again by the assumption that $C \in \mathbb{C}^+$, we have
\begin{align*}
\frac{d}{d \alpha} RE(Y, T(X)) \leq X \cdot CX - Y \cdot CX + \frac{\exp\{ m \alpha \} - 1}{m}  \sum_{j=1}^n X(j) (CX)_j^2.
\end{align*}
Choosing $\bar{C} = \max \left\{\sum X(j) (CX)_j^2 \right\} / m$ and combining with \eqref{ooone} yields the lemma.
\end{proof}

It is now a simple step to show that

\begin{lemma}
\label{convexity_lemma_normalized}
Let $C \in \mathbb{\hat{C}}$. Then, for all $Y \in \mathbb{X}(C)$ and for all $X \in \mathbb{\mathring{X}}(C)$, we have that
\begin{align*}
\forall \alpha > 0 : RE(Y, T(X)) \leq RE(Y, X) - \alpha (Y-X) \cdot CX + \alpha (\exp\{\alpha\} - 1).
\end{align*}
\end{lemma}

\begin{proof}
Using the assumption $C \in \mathbb{\hat{C}}$, we have
\begin{align*}
\frac{1}{m}  \sum_{j=1}^n X(j) (CX)_j^2 \leq \frac{1}{m}  \sum_{j=1}^n X(j) (CX)_j \leq \frac{1}{(CX)_{\max}}  \sum_{j=1}^n X(j) (CX)_j \leq 1.
\end{align*}
Therefore,
\begin{align*}
\frac{d}{d \alpha} RE(Y, T(X)) \leq X \cdot CX - Y \cdot CX + (\exp\{ \alpha \} - 1).
\end{align*}
Combining with \eqref{ooone} yields the lemma.
\end{proof}

We note that \eqref{ooone} can be obtained from Slater's inequality as follows: Let $\hat{X} \equiv T(X)$. Then
\begin{align*}
\frac{\hat{X}(i)}{X(i)} = \frac{\exp\{ \alpha (CX)_i \}}{\sum_{j=1}^n X(j) \exp\{ \alpha (CX)_j \}}.
\end{align*}
Slater's inequality (cf. \citep{Dragomir}) gives
\begin{align*}
\sum_{j=1}^n X (j) \exp\{\alpha (CX)_j\} \leq \exp\left\{ \alpha \frac{\sum_{j=1}^n X(j) (CX)_j \exp\{\alpha (CX)_j\}}{\sum_{j=1}^n X(j) \exp\{\alpha (CX)_j\}} \right\}.
\end{align*}
Combining the previous inequalities, we obtain
\begin{align*}
\frac{\hat{X}(i)}{X(i)} \geq \exp\left\{ \alpha \left( (CX)_i - \frac{\sum_{j=1}^n X(j) (CX)_j \exp\{\alpha (CX)_j\}}{\sum_{j=1}^n X(j) \exp\{\alpha (CX)_j\}} \right) \right\}.
\end{align*}
Taking logarithms
\begin{align*}
\ln(\hat{X}(i)) - \ln(X(i)) \geq \alpha \left( (CX)_i - \frac{\sum_{j=1}^n X(j) (CX)_j \exp\{\alpha (CX)_j\}}{\sum_{j=1}^n X(j) \exp\{\alpha (CX)_j\}} \right).
\end{align*}
Multiplying both sides with $Y(i)$ and summing over $i=1, \ldots, n$
\begin{align*}
RE(Y, X) - RE(Y, \hat{X}) \geq \alpha \left( Y \cdot CX - \frac{\sum_{j=1}^n X(j) (CX)_j \exp\{\alpha (CX)_j\}}{\sum_{j=1}^n X(j) \exp\{\alpha (CX)_j\}} \right).
\end{align*}
Rearranging
\begin{align*}
RE(Y, \hat{X}) - RE(Y, X) \leq \alpha \left( \frac{\sum_{j=1}^n X(j) (CX)_j \exp\{\alpha (CX)_j\}}{\sum_{j=1}^n X(j) \exp\{\alpha (CX)_j\}} - Y \cdot CX \right),
\end{align*}
and, thus,
\begin{align*}
RE(Y, \hat{X}) - RE(Y, X) \leq \alpha \frac{d}{d \alpha} RE(Y, T(X)).
\end{align*}

\subsection{Approximation bounds obtained from the upper relative-entropy bound}

The following lemma is an analogue of \cite[Proposition 8.2.3]{ConvexAnalysis}.

\begin{lemma}
\label{approximation_bound_lemma_weighted_average}
Let $C \in \mathbb{C}^+$, $X^k \equiv T^k(X^0)$, and assume $\alpha_k > 0$, $k = 0, \ldots, K$. Then, for any $\theta > 0$,
\begin{align}
\forall Y \in \mathbb{X}(C) : \frac{1}{A_K} \sum_{k=0}^K \alpha_k (Y - X^k) \cdot CX^k \leq \frac{1}{A_K} \sum_{k=0}^K \alpha_k (\exp\{m \alpha_k\}-1) \bar{C} + \theta\label{xsublime}
\end{align}
where 
\begin{align*}
A_K = \sum_{k = 0}^K \alpha_k,
\end{align*}
$K$ is the smallest integer such that $A_K > RE(Y, X^0)/\theta$ or greater, and $m = \max_{ij} C_{ij}$.
\end{lemma}

\begin{proof}
Assume for the sake of contradiction that \eqref{xsublime} does not hold, that is,
\begin{align}
\frac{1}{A_K} \sum_{k=0}^K \alpha_k (Y - X^k) \cdot CX^k > \frac{1}{A_K} \sum_{k=0}^K \alpha_k (\exp\{m\alpha_k\}-1) \bar{C} + \theta.\label{red}
\end{align}
Invoking Lemma \ref{convexity_lemma_not_normalized},
\begin{align*}
RE(Y, X^{k+1}) \leq RE(Y, X^k) - \alpha (Y-X^k) \cdot CX^k + \alpha (\exp\{m\alpha\} - 1) \bar{C}.
\end{align*}
Summing over $k = 0, \ldots, K$, we obtain
\begin{align*}
RE(Y, X^{K+1}) \leq RE(Y, X^0) - \sum_{k=0}^K \alpha_k (Y - X^k) \cdot CX^k + \sum_{k=0}^K \alpha_k (\exp\{m\alpha_k\}-1) \bar{C}
\end{align*}
and, dividing by $A_K$, we further obtain
\begin{align}
\frac{RE(Y, X^{K+1})}{A_K} \leq \frac{RE(Y, X^0)}{A_K} - \frac{1}{A_K} \sum_{k=0}^K \alpha_k (Y - X^k) \cdot CX^k + \frac{1}{A_K} \sum_{k=0}^K \alpha_k (\exp\{m\alpha_k\}-1) \bar{C}.\label{lalala}
\end{align}
Substituting then \eqref{red} in \eqref{lalala} we obtain
\begin{align*}
\frac{RE(Y, X^{K+1})}{A_K} \leq \frac{RE(Y, X^0)}{A_K} - \theta,
\end{align*}
which implies that
\begin{align*}
RE(Y, X^{K+1}) \leq RE(Y, X^0) - A_K \theta
\end{align*}
and, therefore, that
\begin{align*}
RE(Y, X^0) \geq A_K \theta.
\end{align*}
But this contradicts the previous definition of $K$ and completes the proof.
\end{proof}

\begin{lemma}
\label{SC_lemma}
Assume the learning rate $\alpha_k > 0$ is chosen from round to round such that
\begin{align*}
\lim_{k \rightarrow \infty} \alpha_k = 0 \mbox{ and } \sum_{k = 0}^{\infty} \alpha_k = +\infty. 
\end{align*}
Assume $m > 0$. Then
\begin{align*}
\lim_{K \rightarrow \infty} \frac{1}{A_K} \sum_{k=0}^K \alpha_k (\exp\{m\alpha_k\} - 1) = 0
\end{align*}
\end{lemma}

\begin{proof}
Our proof makes use of the Stolz-Ces\'aro theorem: Let $\{ a_n \}_0^{\infty}$ and $\{ b_n \}_0^{\infty}$ be two sequences of real numbers. Assume $\{ b_n \}_0^{\infty}$ is strictly increasing and approaches $+ \infty$ and
\begin{align*}
\lim_{n \rightarrow \infty} \frac{a_{n+1} - a_n}{b_{n+1} - b_n} = \ell.
\end{align*} 
Then the Stolz-Ces\'aro theorem implies that
\begin{align*}
\lim_{n \rightarrow \infty} \frac{a_n}{b_n} = \ell.
\end{align*}
Coming back to our proof, we let
\begin{align*}
a_n \equiv \sum_{k=0}^n \alpha_k (\exp\{m\alpha_k\} - 1)
\end{align*}
and
\begin{align*}
b_n \equiv A_n.
\end{align*}
Then 
\begin{align*}
\frac{a_{n+1} - a_n}{b_{n+1} - b_n} = \frac{\alpha_{n+1} (\exp\{m\alpha_{n+1}\} - 1)}{a_{n+1}} = \exp\{m\alpha_{n+1}\} - 1.
\end{align*}
Since, by the assumption of the lemma that $\lim_{k \rightarrow \infty} \alpha_k = 0$,
\begin{align*}
\lim_{n \rightarrow \infty} \left\{ \exp\{m\alpha_{n+1}\} - 1 \right\} = 0,
\end{align*}
invoking the Stolz-Ces\'aro theorem, completes the proof.
\end{proof}

\subsection{A lower relative-entropy bound}

\begin{lemma}
\label{lower_relative_entropy_bound}
Let $C \in \mathbb{C}$. Then, for all $Y \in \mathbb{X}(C)$ and for all $X \in \mathbb{\mathring{X}}(C)$, we have that
\begin{align*}
\forall \alpha > 0 : RE(Y, T(X)) \geq RE(Y, X) - \alpha (Y-X) \cdot CX.
\end{align*}
\end{lemma}

\begin{proof}
Lemma \ref{convexity_lemma} and the left-hand-side of the secant inequality \eqref{secant_inequality} imply that
\begin{align*}
RE(Y, T(X)) - RE(Y, X) \geq \alpha \left. \frac{d RE(Y, T(X))}{d \alpha} \right|_{\alpha = 0} = - \alpha (Y-X) \cdot CX
\end{align*}
as claimed.
\end{proof}

\subsection{Approximation bounds obtained from lower relative-entropy bound}

\begin{lemma}
\label{limejuice_lemma}
Let $C \in \mathbb{C}$. Let $X^k \equiv T^k(X^0)$, $k = 1, \ldots K$. For any $X^0 \in \mathbb{\mathring{X}}(C)$,
\begin{align*}
\frac{1}{A_K} \left( \ln(X^{K+1}(i)) - \ln(X^0(i)) \right) \leq (C\bar{X}^K)_i - \frac{1}{A_K} \sum_{k=0}^K \alpha_k X^k \cdot CX^k.
\end{align*}
\end{lemma}

\begin{proof}[First proof of Lemma \ref{limejuice_lemma}]
Lemma \ref{lower_relative_entropy_bound} implies that, for $k=0, \ldots, K$,
\begin{align}
RE(E_i, X^{k+1}) - RE(E_i, X^k) \geq - \alpha_k (E_i - X^k) \cdot CX^k.\label{beachhouse}
\end{align}
Summing over $k=0, \ldots, K$ and dividing by $A_K$ we obtain
\begin{align*}
RE(E_i, X^{K+1}) - RE(E_i, X^0) \geq - \frac{1}{A_K} \sum_{k=0}^K \alpha_k (E_i - X^k) \cdot CX^k.
\end{align*}
Using the definition of relative entropy on the left-hand-side yields 
\begin{align*}
-\ln\left( X^{K+1}(i) \right) + \ln\left(X^0(i)\right) \geq - \frac{1}{A_K} \sum_{k=0}^K \alpha_k (E_i - X^k) \cdot CX^k.
\end{align*}
and negating yields the lemma.
\end{proof}

\begin{proof}[Second proof of Lemma \ref{limejuice_lemma}]
Let $\hat{X} \equiv T(X)$. We then have that
\begin{align*}
\frac{\hat{X}(i)}{X(i)} = \frac{\exp\{ \alpha (CX)_i \}}{\sum_{j=1}^n X(j) \exp\{ \alpha (CX)_j \}}.
\end{align*}
Taking logarithms on both sides, we obtain
\begin{align*}
\ln(\hat{X}(i)) - \ln(X(i)) = \alpha (CX)_i - \ln \left( \sum_{j=1}^n X(j) \exp\{ \alpha (CX)_j \} \right).
\end{align*}
Using Jensen's inequality, we further obtain
\begin{align*}
\ln(\hat{X}(i)) - \ln(X(i)) \leq \alpha ((CX)_i - X \cdot CX).
\end{align*}
We may now write the previous inequality as
\begin{align*}
\ln(X^{k+1}(i)) - \ln(X^k(i)) \leq \alpha_k ((CX^k)_i - X^k \cdot CX^k).
\end{align*}
Summing over $k = 0, \ldots, K$ and dividing by $A_K = \sum_k \alpha_k$,
\begin{align*}
\frac{1}{A_K} \left( \ln(X^{K+1}(i)) - \ln(X^0(i)) \right) \leq (C\tilde{X})_i - \frac{1}{A_K} \sum_{k=0}^K \alpha_k X^k \cdot CX^k,
\end{align*}
as claimed.
\end{proof}

\section{Limit points of sequence of empirical averages are equilibria}
\label{asymptotic_convergence_section}

\subsection{Our main asymptotic convergence result}

\begin{theorem}
\label{asymptotic_convergence_theorem}
Let $C \in \mathbb{C}^+$ and $X^k \equiv T^k(X^0)$, where $X^0 \in \mathbb{\mathring{X}}(C)$. Assume the learning rate $\alpha_k > 0$ is chosen from round to round such that
\begin{align*}
\lim_{k \rightarrow \infty} \alpha_k = 0, \quad  \sum_{k = 0}^{\infty} \alpha_k = +\infty, \mbox{ and } \sum_{k = 0}^{\infty} \alpha_k (\exp\{\alpha_k \}-1) < \infty.
\end{align*}
Then, considering the sequence of weighted empirical averages
\begin{align*}
\left\{ \frac{1}{A_K} \sum_{k=0}^K \alpha_k X^k \equiv \bar{X}^K \right\}_{K=0}^{\infty},
\end{align*}
every limit point of this sequence is a symmetric equilibrium strategy of $C$.
\end{theorem}

\begin{lemma}
\label{main_convergence_lemma}
Under the assumptions of Theorem \ref{asymptotic_convergence_theorem},
\begin{align*}
\forall Y \in \mathbb{X}(C) : \limsup\limits_{K \rightarrow \infty} \left\{ \frac{1}{A_K} \sum_{k=0}^K \alpha_k (Y - X^k) \cdot CX^k \right\} \leq 0.
\end{align*}
\end{lemma}

\begin{proof}
Using Lemma \ref{convexity_lemma_not_normalized}, we obtain that
\begin{align*}
RE(Y, X^{k+1}) \leq RE(Y, X^k) - \alpha_k (Y - X^k) \cdot CX^k + \alpha_k (\exp\{m\alpha_k\} - 1)\bar{C}.
\end{align*}
Summing over $k=0,\ldots,K$ and dividing by $A_K$ we obtain
\begin{align*}
\frac{RE(Y, X^{K+1})}{A_K} \leq \frac{RE(Y, X^0)}{A_K} - \frac{1}{A_K} \sum_{k=0}^K \alpha_k (Y - X^k) \cdot CX^k + \frac{1}{A_K} \sum_{k=0}^K \alpha_k (\exp\{m\alpha_k\} - 1)\bar{C},
\end{align*}
which implies
\begin{align*}
0 \leq \frac{RE(Y, X^0)}{A_K} - \frac{1}{A_K} \sum_{k=0}^K \alpha_k (Y - X^k) \cdot CX^k + \frac{1}{A_K} \sum_{k=0}^K \alpha_k (\exp\{m\alpha_k\} - 1)\bar{C}
\end{align*}
and, rearranging, we obtain
\begin{align*}
\frac{1}{A_K} \sum_{k=0}^K \alpha_k (Y - X^k) \cdot CX^k \leq \frac{RE(Y, X^0)}{A_K} + \frac{1}{A_K} \sum_{k=0}^K \alpha_k (\exp\{m\alpha_k\} - 1)\bar{C}.
\end{align*}
Lemma \ref{SC_lemma} gives that
\begin{align*}
\lim_{K \rightarrow \infty} \left\{ \frac{1}{A_K} \sum_{k=0}^K \alpha_k (\exp\{m\alpha_k\} - 1)\bar{C} \right\} = 0.
\end{align*}
Furthermore, $RE(Y, X^0)$ is finite for all $Y \in \mathbb{X}(C)$. In fact, \citep[Inequality (323)]{Sason-Verdu} implies that
\begin{align*}
RE(Y, X^0) \leq \log \left( \frac{1}{X^0(\min)} \right) \| Y - X^0 \|,
\end{align*}
where $X^0(\min) > 0$ since $X^0$ is interior to the probability simplex by assumption. Letting $\Delta$ be the maximum Euclidean distance between any two points on the simplex, we further obtain
\begin{align*}
RE(Y, X^0) \leq \log \left( \frac{1}{X^0(\min)} \right) \Delta.
\end{align*}
Therefore, since $A_K \rightarrow \infty$ as $K \rightarrow \infty$,
\begin{align*}
\lim_{K \rightarrow \infty} \left\{ \frac{1}{A_K} RE(Y, X^0) \right\} = 0
\end{align*}
and, thus,
\begin{align*}
\limsup\limits_{K \rightarrow \infty} \left\{ \frac{1}{A_K} \sum_{k=0}^K \alpha_k (Y - X^k) \cdot CX^k \right\} \leq 0
\end{align*}
as claimed.
\end{proof}

\begin{lemma}
\label{probability_mass_convergence_2}
Under the assumptions of Theorem \ref{asymptotic_convergence_theorem}, if
\begin{align}
\liminf\limits_{K \rightarrow \infty} \left\{ \frac{1}{A_K} \sum_{k=0}^K \alpha_k (E_i - X^k) \cdot CX^k \right\} < 0,\label{good_assumption}
\end{align}
then
\begin{align*}
\liminf\limits_{K \rightarrow \infty} X^K(i) = 0.
\end{align*}
\end{lemma}

\begin{proof}
Let us assume to the contrary that
\begin{align}
\liminf\limits_{K \rightarrow \infty} X^K(i) = \epsilon > 0.\label{ggg_assumption}
\end{align}
Then, since our map $T$ does not escape the interior of the probability simplex except in the limit, every subsequence is bounded from below by some positive number, say, $c$. Let $\{ X^{K_j} \}_{j = 0}^{\infty}$ be an arbitrary subsequence. From Lemma \ref{limejuice_lemma}, we have
\begin{align*}
\frac{1}{A_K} \left( \ln(X^{K+1}(i)) - \ln(X^0(i)) \right) \leq \frac{1}{A_K} \sum_{k=0}^K \alpha_k (E_i - X^k) \cdot CX^k
\end{align*}
and, since
\begin{align*}
\frac{1}{A_K} \left( \ln(X^{K+1}(i)) - \ln(X^0(i)) \right) \geq \frac{1}{A_K} \ln(X^{K+1}(i)),
\end{align*}
we obtain
\begin{align}
\frac{1}{A_K} \ln(X^{K+1}(i)) \leq \frac{1}{A_K} \sum_{k=0}^K \alpha_k (E_i - X^k) \cdot CX^k,
\end{align}
which implies
\begin{align*}
\forall j = 0, 1, 2, \ldots : \frac{1}{A_{K_j}} \ln(X^{K_j + 1}(i)) \leq \frac{1}{A_{K_j}} \sum_{k=0}^{K_j} \alpha_k (E_i - X^k) \cdot CX^k
\end{align*}
and, thus, since, for all $j$, $X^{K_j}(i) \geq c$, we have
\begin{align}
\forall j = 0, 1, 2, \ldots : \frac{1}{A_{K_j}} \ln(c) \leq \frac{1}{A_{K_j}} \sum_{k=0}^{K_j} \alpha_k (E_i - X^k) \cdot CX^k.\label{bee}
\end{align}
Considering the sequence
\begin{align*}
\left\{\frac{1}{A_{K_j}} \sum_{k=0}^{K_j} \alpha_k (E_i - X^k) \cdot CX^k \right\}_{j = 0}^{\infty},
\end{align*}
let
\begin{align*}
\left\{\frac{1}{A_{K_{j_{\ell}}}} \sum_{k=0}^{K_{j_{\ell}}} \alpha_k (E_i - X^k) \cdot CX^k \right\}_{\ell = 0}^{\infty}
\end{align*}
be a convergent subsequence. Then, invoking \eqref{bee}, we have
\begin{align*}
\frac{1}{A_{K_{j_{\ell}}}} \ln(c) \leq \frac{1}{A_{K_{j_{\ell}}}} \sum_{k=0}^{K_{j_{\ell}}} \alpha_k (E_i - X^k) \cdot CX^k.
\end{align*}
Now as $\ell \rightarrow \infty$, we obtain that the left-hand-size becomes $0$ and thus the right-hand-side is nonnegative, contradicting assumption \eqref{good_assumption}. Therefore, \eqref{ggg_assumption} is false, which completes the proof.
\end{proof}

\begin{lemma}
\label{best_response_lemma}
Let $C \in \mathbb{C}$ and $X^k \equiv T^k(X^0)$, $k =0, 1, \ldots K$. Then
\begin{align*}
X^{K+1} \cdot C\bar{X}^K - \frac{1}{A_K} \sum_{k=0}^K \alpha_k X^k \cdot CX^k \geq \frac{RE(X^{K+1}, X^0)}{A_K} \geq 0.
\end{align*}
\end{lemma}

\begin{proof}
Lemma \ref{lower_relative_entropy_bound} implies that, for all $Y \in \mathbb{X}(C)$,
\begin{align*}
RE(Y, T(X)) \geq RE(Y, X) - \alpha (Y-X) \cdot CX
\end{align*}
which we may write as
\begin{align*}
RE(Y, X^{k+1}) - RE(Y, X^k) \geq - \alpha_k (Y- X^k) \cdot CX^k.
\end{align*}
Summing over $k=0, \ldots, K$ and dividing by $A_K$, we obtain
\begin{align*}
\frac{1}{A_K} \left( RE(Y, X^{K+1}) - RE(Y, X^0) \right) \geq - \frac{1}{A_K} \sum_{k=0}^K \alpha_k (Y- X^k) \cdot CX^k.
\end{align*}
Letting $Y = X^{K+1}$, since $RE(X^{K+1}, X^{K+1}) = 0$, we further obtain
\begin{align*}
- \frac{1}{A_K} RE(X^{K+1}, X^0) \geq - \frac{1}{A_K} \sum_{k=0}^K \alpha_k (X^{K+1}- X^k) \cdot CX^k.
\end{align*}
Rearranging and invoking the property that the relative entropy is a nonnegative function completes the proof.
\end{proof}

\begin{lemma}
\label{limsup_lemma}
Under the assumptions of Theorem \ref{asymptotic_convergence_theorem}, for all $Y \in \mathbb{X}(C)$ and for all $\hat{k} \geq 0$ there exists $0 < \mathsf{C} < \infty$ such that
\begin{align}
\limsup\limits_{K \rightarrow \infty} \left\{ \sum_{k=\hat{k}}^{K} \alpha_k (Y - X^k) \cdot CX^k \right\} \leq \mathsf{C}.\label{fantastic2}
\end{align}
\end{lemma}

\begin{proof}
We show first that
\begin{align}
\forall Y \in \mathbb{X}(C) : \sum_{k=\hat{k}}^K \alpha_k (Y - X^k) \cdot CX^k \leq \sum_{k=\hat{k}}^K \alpha_k (\exp\{\alpha_k\}-1) + RE(Y, X^{\hat{k}}).\label{fantastic}
\end{align}
Assume for the sake of contradiction that \eqref{fantastic} does not hold, that is,
\begin{align}
\sum_{k=\hat{k}}^K \alpha_k (Y - X^k) \cdot CX^k > \sum_{k=\hat{k}}^K \alpha_k (\exp\{\alpha_k\}-1) + RE(Y, X^{\hat{k}}).\label{fantastic_red}
\end{align}
Invoking Lemma \ref{convexity_lemma_normalized},
\begin{align*}
RE(Y, X^{k+1}) \leq RE(Y, X^k) - \alpha_k (Y-X^k) \cdot CX^k + \alpha_k (\exp\{\alpha_k\} - 1).
\end{align*}
Summing over $k = \hat{k}, \ldots, K$, we obtain
\begin{align}
RE(Y, X^{K+1}) \leq RE(Y, X^{\hat{k}}) - \sum_{k=\hat{k}}^K \alpha_k (Y - X^k) \cdot CX^k + \sum_{k=\hat{k}}^K \alpha_k (\exp\{\alpha_k\}-1)\label{fantastic_lalala}
\end{align}
Substituting then \eqref{fantastic_red} in \eqref{fantastic_lalala} we obtain
\begin{align*}
RE(Y, X^{K+1}) < RE(Y, X^{\hat{k}}) - RE(Y, X^{\hat{k}}) = 0,
\end{align*}
which is a contradiction. Thus, \eqref{fantastic} follows. Note now that
\begin{align*}
\sum_{k=\hat{k}}^K \alpha_k (Y - X^k) \cdot CX^k \leq \max\left\{ \sum_{k=\hat{k}}^K \alpha_k (Y - X^k) \cdot CX^k \bigg | Y \in \mathbb{X}(C)\right\}.
\end{align*}
Since 
\begin{align*}
\sum_{k = 0}^{\infty} \alpha_k (\exp\{\alpha_k \}-1) < \infty,
\end{align*}
letting
\begin{align*}
\mathcal{Y}_K \in \arg\max\left\{ \sum_{k=\hat{k}}^K \alpha_k (Y - X^k) \cdot CX^k \bigg | Y \in \mathbb{X}(C) \right\}
\end{align*}
it remains to prove that there exists $\mathsf{C} > 0$ such that
\begin{align}
\forall \hat{k} \geq 0 : \limsup\limits_{K \rightarrow \infty} \left\{ RE(\mathcal{Y}_K, X^{\hat{k}}) \right\} < \mathsf{C}.\label{toprove}
\end{align}
To that end, $\mathcal{\hat{Y}}$ be a limit point of the sequence $\{ \mathcal{Y}_K \}_0^{\infty}$. Then \eqref{fantastic_lalala} implies that
\begin{align}
RE(\mathcal{\hat{Y}}, X^{\hat{k}}) \leq RE(\mathcal{\hat{Y}}, X^{0}) - \sum_{k=0}^{\hat{k}} \alpha_k (\mathcal{\hat{Y}} - X^k) \cdot CX^k + \sum_{k=0}^{\hat{k}} \alpha_k (\exp\{\alpha_k\}-1)\label{999}
\end{align}
and note that the right-hand side is clearly finite for all finite $\hat{k}$. Furthermore, taking the limit as $\hat{k} \rightarrow \infty$ and invoking Lemma \ref{best_response_lemma}, which implies that 
\begin{align*}
\limsup\limits_{\hat{k} \rightarrow \infty} \left\{ \sum_{k=0}^{\hat{k}} \alpha_k (\mathcal{\hat{Y}} - X^k) \cdot CX^k \right\} \geq 0,
\end{align*}
we obtain that the right-hand-side of \eqref{999} remains upper bounded as $\hat{k} \rightarrow \infty$. Since $\mathcal{\hat{Y}}$ is an arbitrary limit point of the sequence $\{ \mathcal{Y}_K \}_0^{\infty}$, this proves \eqref{toprove} and the lemma.
\end{proof}

\begin{lemma}
\label{probability_mass_convergence_7}
Under the assumptions of Theorem \ref{asymptotic_convergence_theorem}, if
\begin{align*}
\liminf\limits_{K \rightarrow \infty} X^K(i) = 0,
\end{align*}
then
\begin{align*}
\lim_{K \rightarrow \infty} \bar{X}^K(i) = 0.
\end{align*}
\end{lemma}

\begin{proof}
Let $\hat{X} \equiv T(X)$. We then have that
\begin{align*}
\frac{\hat{X}(i)}{X(i)} = \frac{\exp\{ \alpha (CX)_i \}}{\sum_{j=1}^n X(j) \exp\{ \alpha (CX)_j \}}.
\end{align*}
Taking logarithms on both sides, we obtain
\begin{align*}
\ln(\hat{X}(i)) - \ln(X(i)) = \alpha (CX)_i - \ln \left( \sum_{j=1}^n X(j) \exp\{ \alpha (CX)_j \} \right).
\end{align*}
Using Jensen's inequality, we further obtain
\begin{align*}
\ln(\hat{X}(i)) - \ln(X(i)) \leq \alpha ((CX)_i - X \cdot CX).
\end{align*}
We may now write the previous inequality as
\begin{align*}
\ln(X^{k+1}(i)) - \ln(X^k(i)) \leq \alpha_k ((CX^k)_i - X^k \cdot CX^k).
\end{align*}
Summing over $k = \hat{k}, \ldots, K$, we obtain
\begin{align*}
\ln(X^{K+1}(i)) - \ln(X^{\hat{k}}(i)) \leq \sum_{k=\hat{k}}^K \alpha_k (E_i - X^k) \cdot CX^k.
\end{align*}
Exponentiating, we obtain
\begin{align*}
X^{K+1}(i) \leq X^{\hat{k}}(i) \exp \left\{ \sum_{k=\hat{k}}^K \alpha_k (E_i - X^k) \cdot CX^k \right\}.
\end{align*}
Therefore,
\begin{align*}
\limsup\limits_{K \rightarrow \infty} X^{K+1}(i) \leq X^{\hat{k}}(i) \limsup\limits_{K \rightarrow \infty} \left\{ \exp \left\{ \sum_{k=\hat{k}}^{K} \alpha_k (E_i - X^k) \cdot CX^k \right\} \right\}.
\end{align*}
Lemma \ref{limsup_lemma} implies that there exists $\mathsf{C} < \infty$ such that
\begin{align*}
\limsup\limits_{K \rightarrow \infty} \left\{ \sum_{k=\hat{k}}^{K} \alpha_k (E_i - X^k) \cdot CX^k \right\} \leq \mathsf{C},
\end{align*}
which implies
\begin{align}
\limsup\limits_{K \rightarrow \infty} X^{K+1}(i) \leq X^{\hat{k}}(i) \exp \left\{ \mathsf{C} \right\}.\label{ccc}
\end{align}
The simple proof of this latter property, namely, that 
\begin{align*}
\limsup\limits_{K \rightarrow \infty} \left\{ \exp \left\{ \sum_{k=\hat{k}}^{K} \alpha_k (E_i - X^k) \cdot CX^k \right\} \right\} = \exp \left\{ \limsup\limits_{K \rightarrow \infty} \left\{ \sum_{k=\hat{k}}^{K} \alpha_k (E_i - X^k) \cdot CX^k \right\} \right\},
\end{align*}
rests on that the $\limsup$ is the supremum of all sub-sequential limits and that, since the exponential function is a monotone increasing function, sub-sequential limits are preserved under the composition of the exponential function with the elements of the aforementioned sequence of partial sums. Taking the $\liminf$ in the right-hand-side of \eqref{ccc}, we, therefore, obtain
\begin{align*}
\limsup\limits_{K \rightarrow \infty} X^{K+1}(i) \leq \liminf\limits_{\hat{k} \rightarrow \infty} \left\{ X^{\hat{k}}(i) \exp \left\{ \mathsf{C} \right\} \right\}.
\end{align*}
Therefore, since, by assumption,
\begin{align*}
\liminf\limits_{\hat{k} \rightarrow \infty} X^{\hat{k}}(i) = 0
\end{align*}
we obtain that
\begin{align*}
\liminf\limits_{\hat{k} \rightarrow \infty} \left\{ X^{\hat{k}}(i) \exp \left\{ \mathsf{C} \right\} \right\} = 0,
\end{align*}
which implies that
\begin{align*}
\limsup\limits_{K \rightarrow \infty} X^{K+1}(i) = 0
\end{align*}
which further implies
\begin{align*}
\lim_{K \rightarrow \infty} X^{K+1}(i) = 0.
\end{align*}
The final step is to apply the Stolz-Ces\'aro theorem to the sequence
\begin{align*}
\left\{ \frac{1}{A_K} \sum_{k=0}^K \alpha_k X^k(i) \right\}_{K=0}^{\infty},
\end{align*}
which implies $\bar{X}^K(i) \rightarrow 0$ as claimed.
\end{proof}

\begin{proof}[Proof of Theorem \ref{asymptotic_convergence_theorem}]
To prove all limit points of the sequence $\left\{ \bar{X}^K \right\}_{K = 0}^{\infty}$ are necessarily equilibrium strategies, let $\{ \bar{X}^{K_j} \}_{j=0}^{\infty}$ be a subsequence that converges to say $\bar{X}$. Then, since
\begin{align*}
\frac{1}{A_{K_j}} \sum_{k=0}^{K_j} \alpha_k E_i \cdot CX^k = E_i \cdot C \left( \frac{1}{A_{K_j}} \sum_{k=0}^{K_j} \alpha_k X^k \right) = E_i \cdot C \bar{X}^{K_j},
\end{align*}
and $E_i \cdot CX$ is a continuous function of $X$, we obtain that
\begin{align}
\forall i \in \mathcal{K}(C): \lim_{j \rightarrow \infty} \left\{ \frac{1}{A_{K_j}} \sum_{k=0}^{K_j} \alpha_k E_i \cdot CX^k \right\} = E_i \cdot C\bar{X}.\label{xyz2}
\end{align}
Lemma \ref{main_convergence_lemma} implies that if $i \in \mathcal{K}(C)$ is such that
\begin{align*}
\liminf\limits_{K \rightarrow \infty} \left\{ \frac{1}{A_K} \sum_{k=0}^K \alpha_k (E_i - X^k) \cdot CX^k \right\} = 0,
\end{align*}
then
\begin{align}
\lim_{K \rightarrow \infty} \left\{ \frac{1}{A_K} \sum_{k=0}^K \alpha_k (E_i - X^k) \cdot CX^k \right\} = 0.\label{xyz1}
\end{align}
Lemmas \ref{probability_mass_convergence_2} and \ref{probability_mass_convergence_7} imply at least one such $i$ exists. Therefore, \eqref{xyz1} and \eqref{xyz2} together imply that the limit
\begin{align*}
\lim_{j \rightarrow \infty} \left\{ \frac{1}{A_{K_j}} \sum_{k=0}^{K_j} \alpha_k X^k \cdot CX^k \right\}
\end{align*}
exists and that for all $i$ such that \eqref{xyz1} holds
\begin{align*}
E_i \cdot C \bar{X} = \lim_{j \rightarrow \infty} \left\{ \frac{1}{A_{K_j}} \sum_{k=0}^{K_j} \alpha_k X^k \cdot CX^k \right\}.
\end{align*}
Lemma \ref{main_convergence_lemma} further implies that
\begin{align*}
\forall i \in \mathcal{K}(C) : \lim_{j \rightarrow \infty} \left\{ \frac{1}{A_{K_j}} \sum_{k=0}^{K_j} \alpha_k (E_{i} - X^k) \cdot CX^k \right\} \leq 0
\end{align*}
and, thus, that
\begin{align*}
E_i \cdot C \bar{X} \leq \lim_{j \rightarrow \infty} \left\{ \frac{1}{A_{K_j}} \sum_{k=0}^{K_j} \alpha_k X^k \cdot CX^k \right\}.
\end{align*}
Moreover, for all $i' \in \mathcal{K}(C)$, such that
\begin{align*}
\lim_{j \rightarrow \infty} \left\{ \frac{1}{A_{K_j}} \sum_{k=0}^{K_j} \alpha_k (E_{i'} - X^k) \cdot CX^k \right\} < 0
\end{align*}
we obtain that
\begin{align*}
E_{i'} \cdot C \bar{X} < \lim_{j \rightarrow \infty} \left\{ \frac{1}{A_{K_j}} \sum_{k=0}^{K_j} \alpha_k X^k \cdot CX^k \right\}.
\end{align*}
and, by Lemmas \ref{probability_mass_convergence_2} and \ref{probability_mass_convergence_7}, that $\bar{X}(i') = 0$. Thus, $\bar{X}$ is a symmetric equilibrium strategy.
\end{proof}

Note that the sequence $\alpha_0=1$ and $\alpha_k = 1/k, k=1, 2, \ldots$ satisfies the assumptions of Theorem \ref{asymptotic_convergence_theorem}. Invoking the Bolzano-Weierstrass theorem along with Theorem \ref{asymptotic_convergence_theorem},  we obtain as an immediate corollary that every symmetric bimatrix game is equipped with a symmetric equilibrium.\\

The following lemma is useful in the sequel.

\begin{lemma}
\label{average_payoff_lemma}
Under the assumptions of Theorem \ref{asymptotic_convergence_theorem},
\begin{align*}
\lim_{K \rightarrow \infty} \left\{ (C\bar{X}^K)_{\max} - \frac{1}{A_K} \sum_{k=0}^K \alpha_k X^k \cdot CX^k \right\} = 0.
\end{align*}
\end{lemma}

\begin{proof}
Note first that Lemma \ref{best_response_lemma} implies that
\begin{align*}
(C\bar{X}^K)_{\max} - \frac{1}{A_K} \sum_{k=0}^K \alpha_k X^k \cdot CX^k \geq 0,
\end{align*}
which further implies that
\begin{align*}
\liminf\limits_{K \rightarrow \infty} \left\{ (C\bar{X}^K)_{\max} - \frac{1}{A_K} \sum_{k=0}^K \alpha_k X^k \cdot CX^k \right\} \geq 0.
\end{align*}
Let us assume for the sake of contradiction that
\begin{align}
\limsup\limits_{K \rightarrow \infty} \left\{ (C\bar{X}^K)_{\max} - \frac{1}{A_K} \sum_{k=0}^K \alpha_k X^k \cdot CX^k \right\} > 0\label{bra}
\end{align}
and let
\begin{align*}
\left\{ (C\bar{X}^{K_j})_{\max} - \frac{1}{A_{K_j}} \sum_{k=0}^{K_j} \alpha_k X^k \cdot CX^k \right\}_{j = 0}^{\infty}
\end{align*}
be a subsequence such that
\begin{align}
\lim_{j \rightarrow \infty} \left\{ (C\bar{X}^{K_j})_{\max} - \frac{1}{A_{K_j}} \sum_{k=0}^{K_j} \alpha_k X^k \cdot CX^k \right\} > 0.\label{bra2}
\end{align}
Furthermore, let $\{ \bar{X}^{K_{j_{\ell}}} \}_{\ell=0}^{\infty}$ be a subsequence that converges to say $\bar{X}$. Then \eqref{bra2} implies that
\begin{align*}
\lim_{\ell \rightarrow \infty} \left\{ (C\bar{X})_{\max} - \frac{1}{A_{K_{j_{\ell}}}} \sum_{k=0}^{K_{j_{\ell}}} \alpha_k X^k \cdot CX^k \right\} > 0,
\end{align*}
which contradicts Lemma \ref{main_convergence_lemma}. This completes the proof.
\end{proof}

\if(0)

Let us now give a second proof of Theorem \ref{asymptotic_convergence_theorem}. We first need a pair of lemmas.

\begin{lemma}
\label{cat}
Under the assumptions of Theorem \ref{asymptotic_convergence_theorem},
\begin{align*}
\lim_{K \rightarrow \infty} \left\{ (C\bar{X}^K)_{\max} - X^{K+1} \cdot C\bar{X}^K \right\} = 0.
\end{align*}
\end{lemma}

\begin{proof}
Lemma \ref{best_response_lemma} implies that
\begin{align*}
X^{K+1} \cdot C\bar{X}^K - \frac{1}{A^K} \sum_{k=0}^K \alpha^k X^k \cdot CX^k \geq \frac{RE(X^{K+1}, X^0)}{A^K}.
\end{align*}
Since $RE(X^{K+1}, X^0)$ is finite as argued in the proof of Lemma \ref{main_convergence_lemma} and $A^K \rightarrow \infty$ as $K \rightarrow \infty$,
\begin{align*}
\lim_{K \rightarrow \infty} \left\{ \frac{RE(X^{K+1}, X^0)}{A^K} \right\} = 0
\end{align*}
and, therefore,
\begin{align}
\liminf\limits_{K \rightarrow \infty} \left\{ X^{K+1} \cdot C\bar{X}^K - \frac{1}{A^K} \sum_{k=0}^K \alpha^k X^k \cdot CX^k \right\} \geq 0.\label{asdlkfsafjhdhfhfd}
\end{align}
Furthermore,
\begin{align*}
\limsup\limits_{K \rightarrow \infty} \left\{ X^{K+1} \cdot C\bar{X}^K - \frac{1}{A^K} \sum_{k=0}^K \alpha^k X^k \cdot CX^k \right\} &\leq \limsup\limits_{K \rightarrow \infty} \left\{ (C\bar{X}^K)_{\max} - \frac{1}{A^K} \sum_{k=0}^K \alpha^k X^k \cdot CX^k \right\}\\
  &= \lim_{K \rightarrow \infty} \left\{ (C\bar{X}^K)_{\max} - \frac{1}{A^K} \sum_{k=0}^K \alpha^k X^k \cdot CX^k \right\} = 0
\end{align*}
by Lemma \ref{average_payoff_lemma}. That is,
\begin{align}
\limsup\limits_{K \rightarrow \infty} \left\{ X^{K+1} \cdot C\bar{X}^K - \frac{1}{A^K} \sum_{k=0}^K \alpha^k X^k \cdot CX^k \right\} \leq 0\label{asdkflhwuereryfh}
\end{align}
Therefore, combining \eqref{asdlkfsafjhdhfhfd} and \eqref{asdkflhwuereryfh}, we obtain
\begin{align}
\lim_{K \rightarrow \infty} \left\{ X^{K+1} \cdot C\bar{X}^K - \frac{1}{A^K} \sum_{k=0}^K \alpha^k X^k \cdot CX^k \right\} = 0.\label{enaaane}
\end{align}
Furthermore, by Lemma \ref{average_payoff_lemma} 
\begin{align}
\lim_{K \rightarrow \infty} \left\{ (C\bar{X}^K)_{\max} - \frac{1}{A_K} \sum_{k=0}^K \alpha_k X^k \cdot CX^k \right\} = 0.\label{dyyyod}
\end{align}
\eqref{enaaane} and \eqref{dyyyod} imply that
\begin{align*}
\lim_{K \rightarrow \infty} \left\{ (C\bar{X}^K)_{\max} - \frac{1}{A_K} \sum_{k=0}^K \alpha_k X^k \cdot CX^k \right\} - \lim_{K \rightarrow \infty} \left\{ X^{K+1} \cdot C\bar{X}^K - \frac{1}{A^K} \sum_{k=0}^K \alpha^k X^k \cdot CX^k \right\} = 0,
\end{align*}
which further implies
\begin{align*}
&\lim_{K \rightarrow \infty} \left\{ (C\bar{X}^K)_{\max} - X^{K+1} \cdot C\bar{X}^K \right\} =\\
&= \lim_{K \rightarrow \infty} \left\{ \left( (C\bar{X}^K)_{\max} - \frac{1}{A_K} \sum_{k=0}^K \alpha_k X^k \cdot CX^k \right) - \left( X^{K+1} \cdot C\bar{X}^K - \frac{1}{A_K} \sum_{k=0}^K \alpha_k X^k \cdot CX^k \right) \right\}=\\
&= \lim_{K \rightarrow \infty} \left\{ (C\bar{X}^K)_{\max} - \frac{1}{A_K} \sum_{k=0}^K \alpha_k X^k \cdot CX^k \right\} - \lim_{K \rightarrow \infty} \left\{ X^{K+1} \cdot C\bar{X}^K - \frac{1}{A^K} \sum_{k=0}^K \alpha^k X^k \cdot CX^k \right\} = 0,
\end{align*}
and completes the proof.
\end{proof}

\begin{lemma}
\label{kanoni}
Under the assumptions of Theorem \ref{asymptotic_convergence_theorem}, if
\begin{align*}
\limsup\limits_{K \rightarrow \infty} \left\{ (C\bar{X}^K)_{\max} - (C\bar{X}^K)_i \right\} > 0
\end{align*}
then
\begin{align*}
\liminf\limits_{K \rightarrow \infty} X^K(i) = 0.
\end{align*}
\end{lemma}

\begin{proof}
We have
\begin{align}
\lim_{K \rightarrow \infty} \left\{ (C\bar{X}^K)_{\max} - X^{K+1} \cdot C\bar{X}^K \right\} &= \lim_{K \rightarrow \infty} \left\{ (C\bar{X}^K)_{\max} - \sum_{i=1}^n X^{K+1}(i) (C\bar{X}^K)_i \right\} =\notag\\
  &= \lim_{K \rightarrow \infty} \left\{ \sum_{i=1}^n X^{K+1}(i) \left( (C\bar{X}^K)_{\max} - (C\bar{X}^K)_i \right) \right\} = 0\label{askdfjjdhfhhhfhf}
\end{align}
Note that, for all $j \in \{1, \ldots, n\}$,
\begin{align*}
\sum_{i=1}^n X^{K+1}(i) \left( (C\bar{X}^K)_{\max} - (C\bar{X}^K)_i \right) \geq X^{K+1}(j) \left( (C\bar{X}^K)_{\max} - (C\bar{X}^K)_j \right)
\end{align*}
and, therefore,
\begin{align}
\limsup\limits_{K \rightarrow \infty} \left\{ \sum_{i=1}^n X^{K+1}(i) \left( (C\bar{X}^K)_{\max} - (C\bar{X}^K)_i \right) \right\} \geq \limsup\limits_{K \rightarrow \infty} \left\{ X^{K+1}(j) \left( (C\bar{X}^K)_{\max} - (C\bar{X}^K)_j \right) \right\}\label{xckvbxafuhasfds}
\end{align}
and
\begin{align}
\liminf\limits_{K \rightarrow \infty} \left\{ \sum_{i=1}^n X^{K+1}(i) \left( (C\bar{X}^K)_{\max} - (C\bar{X}^K)_i \right) \right\} \geq \liminf\limits_{K \rightarrow \infty} \left\{ X^{K+1}(j) \left( (C\bar{X}^K)_{\max} - (C\bar{X}^K)_j \right) \right\}\label{qweoiieulfskdcvnj}
\end{align}
Since by \eqref{askdfjjdhfhhhfhf}
\begin{align*}
\limsup\limits_{K \rightarrow \infty} \left\{ \sum_{i=1}^n X^{K+1}(i) \left( (C\bar{X}^K)_{\max} - (C\bar{X}^K)_i \right) \right\} = \liminf\limits_{K \rightarrow \infty} \left\{ \sum_{i=1}^n X^{K+1}(i) \left( (C\bar{X}^K)_{\max} - (C\bar{X}^K)_i \right) \right\} = 0
\end{align*}
we obtain from \eqref{xckvbxafuhasfds} and \eqref{qweoiieulfskdcvnj} that, for all $j \in \{1, \ldots, n \}$,
\begin{align*}
\limsup\limits_{K \rightarrow \infty} \left\{ X^{K+1}(j) \left( (C\bar{X}^K)_{\max} - (C\bar{X}^K)_j \right) \right\} = \liminf\limits_{K \rightarrow \infty} \left\{ \sum_{i=1}^n X^{K+1}(i) \left( (C\bar{X}^K)_{\max} - (C\bar{X}^K)_i \right) \right\} = 0
\end{align*}
and, thus, that
\begin{align*}
\lim_{K \rightarrow \infty} \left\{ X^{K+1}(j) \left( (C\bar{X}^K)_{\max} - (C\bar{X}^K)_j \right) \right\}  = 0.
\end{align*}
Let $\{ \bar{X}^{K_{\ell}}\}$ be a convergent subsequence. Then
\begin{align*}
0 = \lim_{\ell \rightarrow \infty} \left\{ X^{K_{\ell}+1}(j) \left( (C\bar{X}^{K_{\ell}})_{\max} - (C\bar{X}^{K_{\ell}})_j \right) \right\} &= \liminf\limits_{\ell \rightarrow \infty} \left\{ X^{K_{\ell}+1}(j) \right\} \lim_{\ell \rightarrow \infty} \left\{ (C\bar{X}^{K_{\ell}})_{\max} - (C\bar{X}^{K_{\ell}})_j \right\}\\
  &\geq \liminf\limits_{K \rightarrow \infty} \left\{ X^{K+1}(j) \right\} \lim_{\ell \rightarrow \infty} \left\{ (C\bar{X}^{K_{\ell}})_{\max} - (C\bar{X}^{K_{\ell}})_j \right\}
\end{align*}
and since $\{ \bar{X}^{K_{\ell}}\}$ is arbitrary, we obtain
\begin{align*}
\liminf\limits_{K \rightarrow \infty} \left\{ X^{K+1}(j) \right\} \limsup\limits_{K \rightarrow \infty} \left\{ (C\bar{X}^{K_{\ell}})_{\max} - (C\bar{X}^{K_{\ell}})_j \right\} = 0,
\end{align*}
which implies the lemma.
\end{proof}

\begin{proof}[Second proof of Theorem \ref{asymptotic_convergence_theorem}]
Let $\bar{X}$ be a limit point of the sequence $\{ \bar{X}^K \}$, let $\{\bar{X}^{K_j}\}$ be a subsequence that converges to $\bar{X}$. Lemmas \ref{kanoni} and \ref{probability_mass_convergence_7} imply that for all $i \in \mathcal{C}(\bar{X})$, $(C\bar{X})_i = (C\bar{X})_{\max}$. Therefore, $\bar{X}$ is a symmetric equilibrium strategy. Since $\bar{X}$ is arbitrary, this completes the proof.
\end{proof}

\fi

\section{Weighted empirical averages are approximate equilibria}
\label{equilibrium_approximation_VLR}

\if(0)

\begin{lemma}
Let $C \in \mathbb{C}^+$ and $X^k \equiv T^k(X^0)$, where $X^0 \in \mathbb{\mathring{X}}(C)$. Assume the learning rate $\alpha_k > 0$ is chosen from round to round such that
\begin{align*}
\lim_{k \rightarrow \infty} \alpha_k = 0 \mbox{ and } \sum_{k = 0}^{\infty} \alpha_k = +\infty. 
\end{align*}
Then
\begin{align*}
\liminf\limits_{K \rightarrow \infty} \left\{ \frac{1}{A_K} \sum_{k=0}^K \alpha_k (\bar{X}^K - X^k) \cdot CX^k \right\} \leq 0.
\end{align*}
\end{lemma}

\begin{proof}
Similar to the proof of Lemma \ref{liminf_weighted_average_lemma}.
\end{proof}

\fi

\begin{theorem}
\label{payoff_of_average_strategy_theorem}
Under the assumptions of Theorem \ref{asymptotic_convergence_theorem},
\begin{align*}
\lim_{K \rightarrow \infty} \left\{ \frac{1}{A_K} \sum_{k=0}^K \alpha_k (\bar{X}^K - X^k) \cdot CX^k \right\} = 0.
\end{align*}
\end{theorem}

In the next pair of lemmas, we assume that the assumptions of Theorem \ref{asymptotic_convergence_theorem} hold.

\begin{lemma}
\label{phos}
Let $Y \in \mathbb{X}(C)$. If
\begin{align}
\liminf\limits_{K \rightarrow \infty} \left\{ \frac{1}{A_K} \sum_{k=0}^K \alpha_k (Y - X^k) \cdot CX^k \right\} < 0,\label{foor}
\end{align}
then
\begin{align*}
\exists i \in \mathcal{C}(Y) : \lim_{K \rightarrow \infty} \bar{X}^K(i) = 0.
\end{align*}
\end{lemma}

\begin{proof}
\eqref{foor} implies that
\begin{align*}
\exists i \in \mathcal{C}(Y) : \liminf\limits_{K \rightarrow \infty} \left\{ \frac{1}{A_K} \sum_{k=0}^K \alpha_k (E_i - X^k) \cdot CX^k \right\} < 0
\end{align*}
for if
\begin{align*}
\forall i \in \mathcal{C}(Y) : \liminf\limits_{K \rightarrow \infty} \left\{ \frac{1}{A_K} \sum_{k=0}^K \alpha_k (E_i - X^k) \cdot CX^k \right\} = 0,
\end{align*}
then Lemma \ref{main_convergence_lemma} implies that
\begin{align*}
\forall i \in \mathcal{C}(Y) : \lim_{K \rightarrow \infty} \left\{ \frac{1}{A_K} \sum_{k=0}^K \alpha_k (E_i - X^k) \cdot CX^k \right\} = 0,
\end{align*}
which further implies that
\begin{align*}
\lim_{K \rightarrow \infty} \left\{ \frac{1}{A_K} \sum_{k=0}^K \alpha_k (Y - X^k) \cdot CX^k \right\} = 0.
\end{align*}
Invoking Lemmas \ref{probability_mass_convergence_2} and \ref{probability_mass_convergence_7} completes the proof.
\end{proof}

\begin{lemma}
\label{apotophos}
If $\bar{X}$ is a limit point of $\left\{ \bar{X}^K \right\}$, then
\begin{align*}
\lim_{K \rightarrow \infty} \left\{ \frac{1}{A_K} \sum_{k=0}^K \alpha_k (\bar{X} - X^k) \cdot CX^k \right\} = 0.
\end{align*}
\end{lemma}

\begin{proof}
Lemma \ref{phos} implies that if
\begin{align*}
\liminf\limits_{K \rightarrow \infty} \left\{ \frac{1}{A_K} \sum_{k=0}^K \alpha_k (\bar{X} - X^k) \cdot CX^k \right\} < 0,
\end{align*}
then
\begin{align*}
\exists i \in \mathcal{C}(\bar{X}) : \lim_{K \rightarrow \infty} \bar{X}^K(i) = 0,
\end{align*}
which contradicts the assumption $\bar{X}$ is a limit point. Therefore,
\begin{align*}
\liminf\limits_{K \rightarrow \infty} \left\{ \frac{1}{A_K} \sum_{k=0}^K \alpha_k (\bar{X} - X^k) \cdot CX^k \right\} \geq 0,
\end{align*}
and combining with Lemma \ref{main_convergence_lemma} completes the proof.
\end{proof}

\begin{proof}[Proof of Theorem \ref{payoff_of_average_strategy_theorem}]
Let us assume for the sake of contradiction that
\begin{align*}
\limsup\limits_{K \rightarrow \infty} \left\{ \frac{1}{A_K} \sum_{k=0}^K \alpha_k (\bar{X}^K - X^k) \cdot CX^k \right\} > 0.
\end{align*}
Let
\begin{align*}
\left\{ \frac{1}{A_{K_j}} \sum_{k=0}^{K_j} \alpha_k (\bar{X}^{K_j} - X^k) \cdot CX^k \right\}_{j=0}^{\infty}
\end{align*}
be a subsequence such that
\begin{align}
\lim_{j \rightarrow \infty} \left\{ \frac{1}{A_{K_j}} \sum_{k=0}^{K_j} \alpha_k (\bar{X}^{K_j} - X^k) \cdot CX^k \right\} > 0.\label{flflflx}
\end{align}
Furthermore, let $\left\{ \bar{X}^{K_{j_{\ell}}} \right\}_{\ell = 0}^{\infty}$ be a convergent subsequence. \eqref{flflflx} implies, on one hand, that
\begin{align}
\lim_{\ell \rightarrow \infty} \left\{ \frac{1}{A_{K_{j_{\ell}}}} \sum_{k=0}^{K_{j_{\ell}}} \alpha_k (\bar{X}^{K_{j_{\ell}}} - X^k) \cdot CX^k \right\} > 0.\label{flflfl2x}
\end{align}
On the other hand, let $\bar{X}$ be the limit of $\left\{ \bar{X}^{K_{j_{\ell}}} \right\}_{\ell = 0}^{\infty}$. Then Lemma \ref{apotophos} implies that
\begin{align*}
\lim_{\ell \rightarrow \infty} \left\{ \frac{1}{A_{K_{j_{\ell}}}} \sum_{k=0}^{K_{j_{\ell}}} \alpha_k (\bar{X} - X^k) \cdot CX^k \right\} = 0
\end{align*}
that contradicts \eqref{flflfl2x}. Therefore, 
\begin{align}
\limsup\limits_{K \rightarrow \infty} \left\{ \frac{1}{A_K} \sum_{k=0}^K \alpha_k (\bar{X}^K - X^k) \cdot CX^k \right\} \leq 0.\label{loooovee}
\end{align}
Let us now assume again for the sake of contradiction that
\begin{align*}
\liminf\limits_{K \rightarrow \infty} \left\{ \frac{1}{A_K} \sum_{k=0}^K \alpha_k (\bar{X}^K - X^k) \cdot CX^k \right\} < 0.
\end{align*}
The steps are analogous to the previous case: Let
\begin{align*}
\left\{ \frac{1}{A_{K_j}} \sum_{k=0}^{K_j} \alpha_k (\bar{X}^{K_j} - X^k) \cdot CX^k \right\}_{j=0}^{\infty}
\end{align*}
be a subsequence such that
\begin{align}
\lim_{j \rightarrow \infty} \left\{ \frac{1}{A_{K_j}} \sum_{k=0}^{K_j} \alpha_k (\bar{X}^{K_j} - X^k) \cdot CX^k \right\} < 0.\label{flflfl}
\end{align}
Furthermore, let $\left\{ \bar{X}^{K_{j_{\ell}}} \right\}_{\ell = 0}^{\infty}$ be a convergent subsequence. \eqref{flflfl} implies, on one hand, that
\begin{align}
\lim_{\ell \rightarrow \infty} \left\{ \frac{1}{A_{K_{j_{\ell}}}} \sum_{k=0}^{K_{j_{\ell}}} \alpha_k (\bar{X}^{K_{j_{\ell}}} - X^k) \cdot CX^k \right\} < 0.\label{flflfl2}
\end{align}
On the other hand, let $\bar{X}$ be the limit of $\left\{ \bar{X}^{K_{j_{\ell}}} \right\}_{\ell = 0}^{\infty}$. Then Lemma \ref{apotophos} implies that
\begin{align*}
\lim_{\ell \rightarrow \infty} \left\{ \frac{1}{A_{K_{j_{\ell}}}} \sum_{k=0}^{K_{j_{\ell}}} \alpha_k (\bar{X} - X^k) \cdot CX^k \right\} = 0
\end{align*}
that contradicts \eqref{flflfl2}. Therefore, 
\begin{align*}
\liminf\limits_{K \rightarrow \infty} \left\{ \frac{1}{A_K} \sum_{k=0}^K \alpha_k (\bar{X}^K - X^k) \cdot CX^k \right\} \geq 0
\end{align*}
and combining with \eqref{loooovee} completes the proof.
\end{proof}

\begin{corollary}
\label{approximate_equilibria_corollary}
Under the assumptions of Theorem \ref{asymptotic_convergence_theorem}, for all $\epsilon > 0$, there exists $N > 0$ such that for all $K \geq N$, $\bar{X}^K$ is an $\epsilon$-approximate symmetric equilibrium strategy of $C$.
\end{corollary}

\begin{proof}
Lemma \ref{average_payoff_lemma} implies that
\begin{align}
\forall \theta' > 0 \mbox{ } \exists N'(\theta') > 0 \mbox{ } \forall K \geq N'(\theta') : (C\bar{X}^K)_{\max} - \frac{1}{A_K} \sum_{k=0}^{K} \alpha_k X^k \cdot CX^k < \theta'.\label{uno}
\end{align}
Theorem \ref{payoff_of_average_strategy_theorem} implies that
\begin{align}
\forall \theta'' > 0 \mbox{ } \exists N''(\theta'') > 0 \mbox{ } \forall K \geq N''(\theta'') : \bar{X}^K \cdot C\bar{X}^K - \frac{1}{A_K} \sum_{k=0}^{K} \alpha_k X^k \cdot CX^k > - \theta''.\label{due}
\end{align}
Letting $\theta' = \theta'' = \epsilon/2$ and $N(\epsilon) = \max\{N'(\epsilon/2), N''(\epsilon/2)\}$ in \eqref{uno} and \eqref{due} we obtain that
\begin{align*}
\forall K \geq N(\epsilon) : (C\bar{X}^K)_{\max} - \bar{X}^K \cdot C\bar{X}^K < \epsilon
\end{align*}
and, therefore, that
\begin{align*}
\forall \epsilon > 0 \mbox{ } \exists N(\epsilon) > 0 \mbox{ } \forall K \geq N(\epsilon) : (C\bar{X}^K)_{\max} - \bar{X}^K \cdot C\bar{X}^K < \epsilon
\end{align*}
as claimed.
\end{proof}

\if(0)

\begin{corollary}
Under the assumptions of Theorem \ref{payoff_of_average_strategy_theorem}, if an isolated equilibrium strategy, say $X^*$, is a limit point of the sequence $\left\{ \bar{X}^K \right\}_{K=0}^{\infty}$, then this sequence converges to $X^*$.
\end{corollary}

\begin{proof}
$...$ DON'T KNOW HOW TO PROVE THIS. $...$
\end{proof}

\fi

\section{On the structure of the limit set of weighted empirical averages}
\label{structure_of_limit_set_section}

In this section, we assume that the assumptions of Theorem \ref{asymptotic_convergence_theorem} hold.

\subsection{Some preliminary observations}

\begin{corollary}
\label{elementary_convergence_corollary}
If the sequence $\left\{ \bar{X}^K \right\}$ has a unique limit point, then it converges.
\end{corollary}

\begin{proof}
If a sequence in a compact space has a unique limit point, then it converges.
\end{proof}

\begin{corollary}
If $C$ has a unique symmetric equilibrium, say $X^*$, then $\left\{ \bar{X}^K \right\}$ converges to $X^*$.
\end{corollary}

\begin{proof}
Simple implication of Theorem \ref{asymptotic_convergence_theorem} and Corollary \ref{elementary_convergence_corollary}.
\end{proof}

\subsection{The limit set of the sequence of empirical averages is connected}

\begin{lemma}
\label{consecutive_empirical_averages_lemma}
We have that
\begin{align*}
\lim_{K \rightarrow \infty} RE(\bar{X}^{K}, \bar{X}^{K+1}) = 0,
\end{align*}
which further implies that
\begin{align*}
\lim_{K \rightarrow \infty} \| \bar{X}^{K+1} - \bar{X}^{K} \| = 0,
\end{align*}
where $\| \cdot \|$ is the Euclidean norm.
\end{lemma}

\begin{proof}
Straight algebra gives the following relationship between $\bar{X}^{K+1}$ and $\bar{X}^K$:
\begin{align*}
\bar{X}^{K+1} = \frac{\alpha_{K+1}}{A_{K+1}} X^{K+1} + \frac{A_K}{A_{K+1}} \bar{X}^K.
\end{align*}
Using elementary properties of the relative entropy function, the previous relationship implies that
\begin{align*}
RE(\bar{X}^{K}, \bar{X}^{K+1}) \leq \frac{\alpha_{K+1}}{A_{K+1}} RE(\bar{X}^K, X^{K+1}).
\end{align*}
Furthermore, $A_{K+1} \geq A_K$. Moreover, there exists an iteration $N$ such that for all $K \geq N$, $\alpha_{K+1} < 1$. Thus, for all large enough $K$,
\begin{align}
RE(\bar{X}^{K}, \bar{X}^{K+1}) \leq \frac{1}{A_{K}} RE(\bar{X}^K, X^{K+1}).\label{betweenex}
\end{align}
Now let us recall \eqref{lalala} (in Lemma \ref{approximation_bound_lemma_weighted_average}):
\begin{align*}
\frac{RE(Y, X^{K+1})}{A_K} \leq \frac{RE(Y, X^0)}{A_k} - \frac{1}{A_K} \sum_{k=0}^K \alpha_k (Y - X^k) \cdot CX^k + \frac{1}{A_K} \sum_{k=0}^K \alpha_k (\exp\{m\alpha_k\}-1) \bar{C},
\end{align*}
which holds for any $Y$. Letting $Y = \bar{X}^K$, we obtain
\begin{align*}
\frac{RE(\bar{X}^K, X^{K+1})}{A_K} \leq \frac{RE(\bar{X}^K, X^0)}{A_K} - \frac{1}{A_K} \sum_{k=0}^K \alpha_k (\bar{X}^K - X^k) \cdot CX^k + \frac{1}{A_K} \sum_{k=0}^K \alpha_k (\exp\{m\alpha_k\}-1) \bar{C}.
\end{align*}
Theorem \ref{payoff_of_average_strategy_theorem} gives
\begin{align*}
\lim_{K \rightarrow \infty} \left\{ \frac{1}{A_K} \sum_{k=0}^K \alpha_k (\bar{X}^K - X^k) \cdot CX^k \right\} = 0
\end{align*}
and Lemma \ref{SC_lemma} gives
\begin{align*}
\lim_{K \rightarrow \infty} \frac{1}{A_K} \sum_{k=0}^K \alpha_k (\exp\{m\alpha_k\} - 1) = 0.
\end{align*}
Furthermore, \citep[Inequality (323)]{Sason-Verdu} implies that
\begin{align*}
RE(\bar{X}^K, X^0) \leq \log \left( \frac{1}{X^0(\min)} \right) \| \bar{X}^K - X^0 \|,
\end{align*}
and letting $\Delta$ be the maximum Euclidean distance between any two points on the probability simplex, we further obtain
\begin{align*}
RE(\bar{X}^K, X^0) \leq \log \left( \frac{1}{X^0(\min)} \right) \Delta.
\end{align*}
Therefore,
\begin{align*}
\lim_{K \rightarrow \infty} \frac{1}{A_K} RE(\bar{X}^K, X^0) = 0
\end{align*}
and, thus,
\begin{align*}
\lim_{K \rightarrow \infty} \left\{ \frac{1}{A_{K}} RE(\bar{X}^K, X^{K+1}) \right\} = 0.
\end{align*}
\eqref{betweenex} implies then that
\begin{align*}
\lim_{K \rightarrow \infty} RE(\bar{X}^{K}, \bar{X}^{K+1}) = 0
\end{align*}
as claimed in the first part of the lemma. The second part is a straightforward implication of Pinsker's inequality, yielding
\begin{align*}
RE(\bar{X}^{K}, \bar{X}^{K+1}) \geq \frac{1}{2} \|\bar{X}^{K+1} - \bar{X}^{K}\|^2
\end{align*}
and, thus,
\begin{align*}
\lim_{K \rightarrow \infty} \| \bar{X}^{K+1} - \bar{X}^{K} \| = 0,
\end{align*}
as claimed.
\end{proof}

\begin{theorem}
\label{limitpoint_connectedness_theorem}
The limit set of the sequence $\left\{ \bar{X}^K \right\}$ is connected.
\end{theorem}


\begin{proof}
Let us assume for the sake of contradiction that the limit set, say, $\mathbb{\bar{X}}^*$ is not connected. Then there exist two disjoint closed subsets of $\mathbb{X}(C)$, say $A$ and $B$, such that the complement of their union, call it $\Omega$, does not intersect $\mathbb{\bar{X}}^*$ (that is, $\mathbb{\bar{X}}^* \cap \Omega = \emptyset$). Since $\mathbb{X}(C)$ is compact, $A$ and $B$ are also compact. Therefore, there exists $\epsilon > 0$ such that
\begin{align}
\forall X \in A \mbox{ } \forall Y \in B : \| X - Y \| \geq \epsilon.\label{battleforatlas}
\end{align}
Since $\Omega$ is such that it does not contain any limit point of $\left\{ \bar{X}^K \right\}$ it contains at most a finite number of points of $\left\{ \bar{X}^K \right\}$. Therefore, there exists an infinite set of indices, say $\mathcal{K}$ such that 
\begin{align*}
\forall K \in \mathcal{K} : \bar{X}^K \in A \mbox{ and } \bar{X}^{K+1} \in B.
\end{align*}
But then, for all $K \in \mathcal{K}$, \eqref{battleforatlas} implies that
\begin{align*}
\| \bar{X}^{K+1} - \bar{X}^K \| \geq \epsilon,
\end{align*}
which contradicts, however, Lemma \ref{consecutive_empirical_averages_lemma} and, thus, the assumption $\mathbb{\bar{X}}^*$ is not connected.
\end{proof}

As straightforward implications of Theorem \ref{limitpoint_connectedness_theorem} we have:

\begin{corollary}
If $\left\{ \bar{X}^K \right\}$ has an isolated equilibrium, say $X^*$, as a limit point, it converges to $X^*$.
\end{corollary}

\begin{corollary}
If $C$ has a finite number of symmetric equilibria, then the sequence $\left\{ \bar{X}^K \right\}$ converges.
\end{corollary}

Note that every {\em non-degenerate game} has a finite number of equilibria.

\subsection{The limit set of averages lies on a convex equilibrium polytope}

\begin{lemma}
\label{limit_points_are_best_responses_to_each_other}
Every limit point of $\left\{ \bar{X}^K \right\}$ is a best response to every other limit point.
\end{lemma}

\begin{proof}
Let $\left\{ \bar{X}^{K_j} \right\}_{j=0}^{\infty}$ be a subsequence that converges to $\bar{X}$ and suppose there exists a second limit point $\bar{X}' \neq \bar{X}$. Theorem \ref{asymptotic_convergence_theorem} implies, on one hand, that
\begin{align*}
(C\bar{X})_{\max} = \lim_{j \rightarrow \infty} \left\{ \frac{1}{A_{K_j}} \sum_{k=0}^{K_j} \alpha_k X^k \cdot CX^k \right\}
\end{align*}
and, on the other, Lemma \ref{apotophos} implies that
\begin{align*}
\bar{X}' \cdot C\bar{X} = \lim_{j \rightarrow \infty} \left\{ \frac{1}{A_{K_j}} \sum_{k=0}^{K_j} \alpha_k X^k \cdot CX^k \right\}.
\end{align*}
Combining the previous equalities completes the proof.
\end{proof}

\begin{theorem}
\label{equilibrium_polytope_limitset_theorem}
The limit set of $\left\{ \bar{X}^K \right\}_{K=0}^{\infty}$ is a connected subset of a convex polytope of equilibria.
\end{theorem}

\begin{proof}
Lemma \ref{equilibrium_set_affine_hull_lemma} implies that the affine hull of a set of equilibria that are best responses to each other is a convex equilibrium polytope. Thus, Lemma \ref{limit_points_are_best_responses_to_each_other} implies that the limit set of $\left\{ \bar{X}^K \right\}_{K=0}^{\infty}$ is a subset of such equilibrium polytope. Theorem \ref{limitpoint_connectedness_theorem} further implies that this limit set is connected.
\end{proof}

Note that \citep[Lemma 1, Theorem 2]{Jansen} relates our result to {\em maximal Nash subsets}.

\section{Exploiting side information in equilibrium computation}
\label{side_information}

In this final section, we explore further the information the sequence of iterates provides.  We assume the assumptions of Theorem \ref{asymptotic_convergence_theorem} hold.

\subsection{Recurrent clean separation using average probability masses}

The results we have insofar obtained imply that for any approximation error, say $\epsilon$, the empirical average is an $\epsilon$-approximate equilibrium in a finite number of steps. In this section, we take this result further by exploiting side information in the equilibrium computation process. One technique that computes an equilibrium in a finite number of steps can be derived from the following theorem:

\begin{theorem}
\label{recurrent_average_probability_masses}
Let $X^*$ be an arbitrary limit point of the sequence $\left\{ \bar{X}^K \right\}$. Let 
\begin{align*}
\mathcal{K}^+(X^*) = \left\{ i \in \mathcal{K}(C) | X^*(i) > 0 \right\}
\end{align*}
and $m = | \mathcal{K}^+(X^*) |$. If we rank pure strategies according to their respective (average) probability mass (as in $\bar{X}^K$), the top-$m$ strategies identify the pure strategies supporting $X^*$ infinitely often.
\end{theorem}

\begin{proof}
Let $\left\{ \bar{X}^{K_j} \right\}$ be a subsequence that converges to $X^*$. If $i \not\in \mathcal{K}^+(X^*)$, then $\left\{ \bar{X}^{K_j}(i) \right\}$ converges to zero. Therefore, for all $\epsilon > 0$, there exists $\ell > 0$ such that for all $j \geq \ell$, $\bar{X}^{K_j}(i) < \epsilon$. Furthermore, if $i \in \mathcal{K}^+(X^*)$, then for all $\epsilon > 0$, there exists $\ell > 0$ such that for all $j \geq \ell$, $\bar{X}^{K_j}(i) > X^*(i) - \epsilon$. If $\epsilon$ is sufficiently small and $\ell$ is sufficiently large, $X^*$'s support is at the top of the ranking.
\end{proof}

As noted earlier we can check in polynomial time whether a subset of pure strategies support an equilibrium. Thus, waiting for the sequence of empirical averages to evolve far enough, an (equilibrium) limit point can be computed efficiently. An interesting question then is, starting from an element of the convex polytope of equilibria where the limit set of $\left\{ \bar{X}^K \right\}$ lies, whether the entire polytope can be computed in polynomial time. We leave this as an open question for the future.

\subsection{Recurrent clean separation using payoffs}

\begin{lemma}
\label{equalizer_supporting_lemma}
Let $X^*$ be a limit point of $\left\{ \bar{X}^K \right\}$ and let the subsequence $\left\{ \bar{X}^{K_j} \right\}$ be such that 
\begin{align}
\lim_{j \rightarrow \infty} \bar{X}^{K_j} = X^*.\label{firestone1}
\end{align}
Then
\begin{align}
\lim_{j \rightarrow \infty} \left\{ \frac{1}{A_{K_j}} \sum_{k=0}^{K_j} \alpha_k (E_i - X^k) \cdot CX^k \right\} = 0\label{firestone}
\end{align}
implies that $E_i$ is a best response to $X^*$.
\end{lemma}

\begin{proof}
\eqref{firestone1} and \eqref{firestone} imply that
\begin{align*}
(CX^*)_i = \lim_{j \rightarrow \infty} \left\{ \frac{1}{A_{K_{j}}} \sum_{k=0}^{K_{j}} \alpha_k X^k \cdot CX^k \right\}
\end{align*}
and Theorem \ref{asymptotic_convergence_theorem} further implies that
\begin{align*}
(CX^*)_{\max} = \lim_{j \rightarrow \infty} \left\{ \frac{1}{A_{K_{j}}} \sum_{k=0}^{K_{j}} \alpha_k X^k \cdot CX^k \right\}.
\end{align*}
Therefore, $(CX^*)_i = (CX^*)_{\max}$, and, thus, $E_i$ is a best response to $X^*$ as claimed.
\end{proof}

\begin{theorem}
\label{equalizer_separation_theorem}
Let $X^*$ be a limit point of $\left\{ \bar{X}^K \right\}$. Let us rank pure strategies according to their respective payoff $(C\bar{X}^K)_i, i =1, \ldots, n$. Then there exists $m \in \{1, \ldots, n\}$ such that the top-$m$ pure strategies form a carrier where $X^*$ is an equilibrium subequalizer infinitely often.
\end{theorem}

\begin{proof}
Let $\left\{ \bar{X}^{K_j} \right\}$ be a subsequence that converges to $X^*$. Lemma \ref{main_convergence_lemma} implies that
\begin{align*}
\forall i \in \mathcal{K}(C) : \lim_{j \rightarrow \infty} \left\{ \frac{1}{A_{K_j}} \sum_{k=0}^{K_j} \alpha_k (E_i - X^k) \cdot CX^k \right\} \leq 0.
\end{align*}
Let $i \in \mathcal{K}(C)$ be such that
\begin{align}
\lim_{j \rightarrow \infty} \left\{ \frac{1}{A_{K_j}} \sum_{k=0}^{K_j} \alpha_k (E_i - X^k) \cdot CX^k \right\} = 0.\label{eenaa}
\end{align}
Let $i' \in \mathcal{K}(C)$ be such that
\begin{align}
\lim_{j \rightarrow \infty} \left\{ \frac{1}{A_{K_j}} \sum_{k=0}^{K_j} \alpha_k (E_{i'} - X^k) \cdot CX^k \right\} = \epsilon_{i'} < 0.\label{ddyoo}
\end{align}
\eqref{ddyoo} implies $X^*(i') = 0$ (cf. Lemmas \ref{probability_mass_convergence_2} and \ref{probability_mass_convergence_7}). Therefore, the pure strategies that satisfy \eqref{eenaa} include the carrier of $X^*$. Lemma \ref{equalizer_supporting_lemma} further implies that all pure strategies that satisfy \eqref{eenaa} are best responses to $X^*$. Therefore, $X^*$ is an equilibrium subequalizer of the pure strategies that satisfy \eqref{eenaa}. Let $i$ be such that \eqref{eenaa} is satisfied. Then, for all $\epsilon > 0$, there exists $N_j$ such that for all $K_j \geq N_j$,
\begin{align*}
\frac{1}{A_{K_j}} \sum_{k=0}^{K_j} \alpha_k (E_i - X^k) \cdot CX^k > -\epsilon.
\end{align*}
Let $i'$ be such that \eqref{ddyoo} is satisfied. Then, for all $\epsilon > 0$, there exists $N_j$ such that for all $K_j \geq N_j$,
\begin{align*}
\frac{1}{A_{K_j}} \sum_{k=0}^{K_j} \alpha_k (E_{i'} - X^k) \cdot CX^k < \epsilon_{i'} + \epsilon.
\end{align*}
Choosing $\epsilon$ small enough and $N_j$ large enough we have the promised recurring clean separation.
\end{proof}

As noted earlier (in Section \ref{christmastree}), given a set of pure strategies, a corresponding equilibrium subequalizer (if it exists) can be computed in polynomial time using linear programming.

\subsection{Recurrent separation using the probability masses of the iterates}

\begin{lemma}
\label{joint_evolution_lemma}
If $X^0$ is the uniform strategy, $\forall i, j \in \mathcal{K}(C)$, $(C \bar{X}^K)_i - (C \bar{X}^K)_j$ and $X^{K+1}(i) - X^{K+1}(j)$ are either both zero or have the same sign.
\end{lemma}

\begin{proof}
Let $T(X) \equiv \hat{X}$. Then
\begin{align*}
\frac{\hat{X}(i)}{\hat{X}(j)} = \frac{X(i)}{X(j)} \exp\{\alpha ((CX)_i - (CX)_j)\}
\end{align*}
and taking logarithms on both sides
\begin{align*}
\ln \left( \frac{\hat{X}(i)}{\hat{X}(j)} \right) = \ln \left( \frac{X(i)}{X(j)} \right) + \alpha ((CX)_i - (CX)_j).
\end{align*}
We may rewrite the previous equation as
\begin{align*}
\ln \left( \frac{X^{k+1}(i)}{X^{k+1}(j)} \right) = \ln \left( \frac{X^k(i)}{X^k(j)} \right) + \alpha_k ((CX^k)_i - (CX^k)_j).
\end{align*}
Summing over $k = 0, \ldots K$, we obtain
\begin{align*}
\ln \left( \frac{X^{K+1}(i)}{X^{K+1}(j)} \right) = \ln \left( \frac{X^0(i)}{X^0(j)} \right) + \sum_{k=0}^K \alpha_k ((CX^k)_i - (CX^k)_j)
\end{align*}
Assuming iterations start from the uniform distribution, we obtain
\begin{align*}
\ln \left( \frac{X^{K+1}(i)}{X^{K+1}(j)} \right) = \sum_{k=0}^K \alpha_k ((CX^k)_i - (CX^k)_j)
\end{align*}
and, dividing by $A_K$, we further obtain
\begin{align*}
\frac{1}{A_K} \ln \left( \frac{X^{K+1}(i)}{X^{K+1}(j)} \right) = (E_i - E_j) \cdot C\bar{X}^K.
\end{align*}
The statement of the lemma follows immediately from the latter equation.
\end{proof}

\begin{corollary}
\label{probability_mass_ranking_corollary}
Let $X^0$ be the uniform strategy and $X^*$ a limit point of $\left\{ \bar{X}^K \right\}$. Let us rank pure strategies according to their respective probability masses $X^K(i), i =1, \ldots, n$. Then there exists $m \in \{1, \ldots, n\}$ such that the top-$m$ pure strategies form a carrier where $X^*$ is an equilibrium subequalizer infinitely often.
\end{corollary}

\begin{proof}
Immediate implication of Theorem \ref{equalizer_separation_theorem} and Lemma \ref{joint_evolution_lemma}.
\end{proof}

\subsection{Persistent clean separation using best responses}

We first show a ``weak recurrence'' phenomenon involving best responses as Theorem \ref{recurring_best_responses_phenomenon} below. In the proof of this theorem, we need the following lemma:

\begin{lemma}
\label{damaskino}
Let $C \in \mathbb{C}$, $X^k \equiv T^k(X^0)$, $k = 0, 1, \ldots$, and $X^0 \in \mathbb{\mathring{X}(C)}$. Then,
\begin{align*}
\forall i \in \mathcal{K}(C) \mbox{ } \exists c > 0 \mbox{ } \forall K \geq 0 : (C\bar{X}^K)_i - (C\bar{X}^K)_{\max} \geq \frac{1}{A_K} \ln \left( c \right) + \frac{1}{A_K} \ln \left( X^{K+1}(i) \right)
\end{align*}
\end{lemma}

\begin{proof}
Let $\hat{X} \equiv T(X)$. We then have
\begin{align*}
\frac{\hat{X}(i)}{\hat{X}(j)} = \frac{X(i)}{X(j)} \exp\{\alpha ((CX)_i - (CX)_j)\}.
\end{align*}
Taking logarithms,
\begin{align*}
\ln \left( \frac{\hat{X}(i)}{\hat{X}(j)} \right) = \ln \left( \frac{X(i)}{X(j)} \right) + \alpha ((CX)_i - (CX)_j).
\end{align*}
We may rewrite the previous equation as
\begin{align*}
\ln \left( \frac{X^{k+1}(i)}{X^{k+1}(j)} \right) = \ln \left( \frac{X^k(i)}{X^k(j)} \right) + \alpha_k ((CX^k)_i - (CX^k)_j).
\end{align*}
Summing over $k = 0, \ldots K$, we obtain
\begin{align*}
\ln \left( \frac{X^{K+1}(i)}{X^{K+1}(j)} \right) = \ln \left( \frac{X^0(i)}{X^0(j)} \right) + \sum_{k=0}^K \alpha_k ((CX^k)_i - (CX^k)_j).
\end{align*}
Dividing by $A_K$, we thus obtain
\begin{align*}
\frac{1}{A_K} \ln \left( \frac{X^{K+1}(i)}{X^{K+1}(j)} \right) = \frac{1}{A_K} \ln \left( \frac{X^0(i)}{X^0(j)} \right) + (E_i - E_j) \cdot C\bar{X}^K.
\end{align*}
The previous equation implies that
\begin{align*}
(E_i - E_j) \cdot C\bar{X}^K \geq \frac{1}{A_K} \ln \left( \frac{X^0(j)}{X^0(i)} \right) + \frac{1}{A_K} \ln \left( X^{K+1}(i) \right),
\end{align*}
which further implies
\begin{align*}
(E_i - E_j) \cdot C\bar{X}^K \geq \frac{1}{A_K} \ln \left( \frac{X^0(\min)}{X^0(\max)} \right) + \frac{1}{A_K} \ln \left( X^{K+1}(i) \right).
\end{align*}
Letting
\begin{align*}
c = \frac{X^0(\min)}{X^0(\max)},
\end{align*}
since $j$ is arbitrary, we obtain
\begin{align*}
(C\bar{X}^K)_i - (C\bar{X}^K)_{\max} \geq \frac{1}{A_K} \ln \left( c \right) + \frac{1}{A_K} \ln \left( X^{K+1}(i) \right)
\end{align*}
as claimed.
\end{proof}

\begin{theorem}
\label{recurring_best_responses_phenomenon}
If the pure strategy $E_i$ supports a limit point of the sequence $\left\{ \bar{X}^K \right\}_{K=0}^{\infty}$, then, for all $\epsilon > 0$, $E_i$ is an $\epsilon$-approximate best response to $\bar{X}^K, K = 0, 1, \ldots$ infinitely often.
\end{theorem}

\begin{proof}
Let us assume that there is $\epsilon > 0$ such that $E_i$ is an $\epsilon$-approximate best response to $\bar{X}^K, K = 0, 1, \ldots$ only a finite number of times. Then 
\begin{align*}
\exists N \geq 0 \mbox{ } \forall K \geq N : (C\bar{X}^K)_i - (C\bar{X}^K)_{\max} < - \epsilon.
\end{align*}
Lemma \ref{damaskino} implies then that, for some $c > 0$,
\begin{align*}
X^{K+1}(i) \leq c \exp\left\{ - \epsilon A_K \right\},
\end{align*}
which further implies
\begin{align*}
\limsup\limits_{K \rightarrow \infty} X^{K+1}(i) \leq c \exp\left\{ - \epsilon \left( \lim_{K \rightarrow \infty} A_K \right) \right\},
\end{align*}
Since $A_K \rightarrow \infty$ as $K \rightarrow \infty$, we obtain $X^{K+1}(i) \rightarrow 0$ as $K \rightarrow \infty$ and the Stolz-Ces\'aro theorem implies $\bar{X}^{K}(i) \rightarrow 0$ as $K \rightarrow \infty$. Therefore, $E_i$ cannot support a limit point of $\left\{ \bar{X}^K \right\}_{K=0}^{\infty}$.
\end{proof}

In the sequel, we refine and strengthen Theorem \ref{recurring_best_responses_phenomenon}.

\begin{lemma}
\label{love1}
If
\begin{align}
\limsup\limits_{K \rightarrow \infty} \left\{ \frac{1}{A_K} \sum_{k = 0}^K \alpha_k (E_i - X^k) \cdot CX^k \right\} = - \epsilon' < 0,\label{blackfriday}
\end{align}
then $\exists \hat{\epsilon} > 0$ and $N > 0$ such that, $\forall K \geq N$, $E_i$ is not a $\hat{\epsilon}$-approximate best response to $\bar{X}^K$. 
\end{lemma}

\begin{proof}
\eqref{blackfriday} implies that there exists there exists $\epsilon < \epsilon'$ and an iteration $N$ such that, for all $K \geq N$,
\begin{align}
(C\bar{X}^K)_i - \frac{1}{A_K} \sum_{k = 0}^K \alpha_k X^k \cdot CX^k < \epsilon - \epsilon'.\label{saturday1}
\end{align}
By Lemma \ref{best_response_lemma}, we have that, for all $K \geq 0$,
\begin{align*}
X^{K+1} \cdot C\bar{X}^K \geq \frac{1}{A_K} \sum_{k=0}^K \alpha_k X^k \cdot CX^k,
\end{align*}
and, therefore, that
\begin{align}
(C\bar{X}^K)_{\max} \geq \frac{1}{A_K} \sum_{k=0}^K \alpha_k X^k \cdot CX^k.\label{saturday2}
\end{align}
Combining \eqref{saturday1} and \eqref{saturday2} we obtain
\begin{align*}
(C\bar{X}^K)_i - (C\bar{X}^K)_{\max} < \epsilon - \epsilon' \equiv - \hat{\epsilon}.
\end{align*}
Thus, $\exists N$ such that $\forall K \geq N$, $E_i$ is not a $\hat{\epsilon}$-approximate best response to $\bar{X}^K$ as claimed.
\end{proof}

\begin{lemma}
\label{love2}
If
\begin{align*}
\limsup\limits_{K \rightarrow \infty} \left\{ \frac{1}{A_K} \sum_{k=0}^K \alpha_k (E_i - X^k) \cdot CX^k \right\} = 0,
\end{align*}
then there exists a limit point, say $\bar{X}$, of $\left\{ \bar{X}^K \right\}$ such that $E_i$ is a best response to $\bar{X}$.
\end{lemma}

\begin{proof}
Let
\begin{align*}
\left\{ \frac{1}{A_{K_j}} \sum_{k=0}^{K_j} \alpha_k (E_i - X^k) \cdot CX^k \right\}_{j=0}^{\infty}
\end{align*}
be a subsequence such that
\begin{align}
\lim_{j \rightarrow \infty} \left\{ \frac{1}{A_{K_j}} \sum_{k=0}^{K_j} \alpha_k (E_i - X^k) \cdot CX^k \right\} = 0.\label{huaway}
\end{align}
Furthermore, let $\left\{ \bar{X}^{K_{j_{\ell}}} \right\}_{\ell = 0}^{\infty}$ be a convergent subsequence. Then \eqref{huaway} implies that
\begin{align}
\lim_{\ell \rightarrow \infty} \left\{ \frac{1}{A_{K_{j_{\ell}}}} \sum_{k=0}^{K_{j_{\ell}}} \alpha_k (E_i - X^k) \cdot CX^k \right\} = 0.\label{huaway2}
\end{align}
Let $\bar{X}$ be the limit of $\left\{ \bar{X}^{K_{j_{\ell}}} \right\}_{\ell = 0}^{\infty}$. \eqref{huaway2} implies that
\begin{align*}
(C\bar{X})_i = \lim_{\ell \rightarrow \infty} \left\{ \frac{1}{A_{K_{j_{\ell}}}} \sum_{k=0}^{K_{j_{\ell}}} \alpha_k X^k \cdot CX^k \right\}
\end{align*}
and Theorem \ref{asymptotic_convergence_theorem} implies that
\begin{align*}
(C\bar{X})_i = \bar{X} \cdot C\bar{X}.
\end{align*}
Therefore, since $\bar{X}$ is an equilibrium strategy, $E_i$ is a best response to $\bar{X}$.
\end{proof}

\begin{lemma}
\label{lovestar}
$E_i$ is a best response to some limit point $\bar{X}$ of $\left\{ \bar{X}^K \right\}$ if and only if
\begin{align*}
\limsup\limits_{K \rightarrow \infty} \left\{ \frac{1}{A_K} \sum_{k=0}^K \alpha_k (E_i - X^k) \cdot CX^k \right\} = 0.
\end{align*}
\end{lemma}

\begin{proof}
If
\begin{align*}
\limsup\limits_{K \rightarrow \infty} \left\{ \frac{1}{A_K} \sum_{k=0}^K \alpha_k (E_i - X^k) \cdot CX^k \right\} = 0,
\end{align*}
Lemma \ref{love2} implies there exists a limit point, say $\bar{X}$, of $\left\{ \bar{X}^K \right\}$ such that $E_i$ is a best response to $\bar{X}$. If
\begin{align*}
\limsup\limits_{K \rightarrow \infty} \left\{ \frac{1}{A_K} \sum_{k=0}^K \alpha_k (E_i - X^k) \cdot CX^k \right\} < 0,
\end{align*}
Lemma \ref{love1} implies that, for some $\hat{\epsilon} > 0$ and $N > 0$, $E_i$ is not a $\hat{\epsilon}$-approximate best response to $\bar{X}^K$ for all $K \geq N$. Thus, $E_i$ is not a $\hat{\epsilon}$-approximate best response to any of the limit points of $\bar{X}^K$.
\end{proof}

\if(0)

\begin{theorem}
\label{love3}
If, for some $N > 0$, $E_i$ is a best response to $\bar{X}^K$ for all $K \geq N$, then it is a best response to a limit point of the sequence $\left\{ \bar{X}^K \right\}$.
\end{theorem}

\begin{proof}
By Lemma \ref{main_convergence_lemma},
\begin{align*}
\limsup\limits_{K \rightarrow \infty} \left\{ \frac{1}{A_K} \sum_{k=0}^K \alpha_k (E_i - X^k) \cdot CX^k \right\} \leq 0.
\end{align*}
Lemma \ref{love1} implies that if, for some $N > 0$, $E_i$ is a best response to $\bar{X}^K$ for all $K \geq N$, then
\begin{align*}
\limsup\limits_{K \rightarrow \infty} \left\{ \frac{1}{A_K} \sum_{k=0}^K \alpha_k (E_i - X^k) \cdot CX^k \right\} = 0.
\end{align*}
Then invoking Lemma \ref{love2} completes the proof.
\end{proof}

\fi

\if(0)

\begin{lemma}
\label{some_nice_lemma}
\begin{align}
\limsup\limits_{K \rightarrow \infty} \left\{ \frac{1}{A_K} \sum_{k=0}^K \alpha_k (E_i - X^k) \cdot CX^k \right\} = 0\label{you}
\end{align}
implies that
\begin{align}
\limsup\limits_{K \rightarrow \infty} \left\{ \frac{1}{A_K} \left( RE(E_i, X^0) - RE(E_i, X^{K+1}) \right) \right\} \geq 0.\label{you2}
\end{align}
\end{lemma}

\begin{proof}
Invoking Lemma \ref{convexity_lemma_not_normalized},
\begin{align*}
RE(E_i, X^{k+1}) \leq RE(E_i, X^k) - \alpha (E_i-X^k) \cdot CX^k + \alpha (\exp\{m\alpha\} - 1) \bar{C}.
\end{align*}
Summing over $k = 0, \ldots, K$, we obtain
\begin{align*}
RE(E_i, X^{K+1}) \leq RE(E_i, X^0) - \sum_{k=0}^K \alpha_k (E_i - X^k) \cdot CX^k + \sum_{k=0}^K \alpha_k (\exp\{m\alpha_k\}-1) \bar{C}
\end{align*}
and, dividing by $A_K$, we further obtain
\begin{align}
\frac{RE(E_i, X^{K+1})}{A_K} \leq \frac{RE(E_i, X^0)}{A_K} - \frac{1}{A_K} \sum_{k=0}^K \alpha_k (E_i - X^k) \cdot CX^k + \frac{1}{A_K} \sum_{k=0}^K \alpha_k (\exp\{m\alpha_k\}-1) \bar{C}.\label{youu}
\end{align}
By Lemma \ref{SC_lemma}, we have
\begin{align}
\lim_{K \rightarrow \infty} \frac{1}{A_K} \sum_{k=0}^K \alpha_k (\exp\{m\alpha_k\} - 1) = 0.\label{youuu}
\end{align}
Then \eqref{you}, \eqref{youu}, and \eqref{youuu} imply \eqref{you2}.
\end{proof}

\fi

\begin{theorem}
\label{persistent_best_responses}
For all $\epsilon > 0$, there exists an iteration $N$ such that, for all $K \geq N$, every pure strategy that supports a limit point of $\left\{ \bar{X}^K \right\}$ is an $\epsilon$-approximate best response to $\bar{X}^K$ and every limit point of $\left\{ \bar{X}^K \right\}$ is a well-supported $\epsilon$-approximate best response to $\bar{X}^K$.
\end{theorem}

\begin{proof}
Lemma \ref{average_payoff_lemma} implies on one hand that
\begin{align}
\forall \epsilon > 0 \mbox{ } \exists N > 0 \mbox{ } \forall K \geq N : (C\bar{X}^K)_{\max} - \frac{1}{A_K} \sum_{k=0}^{K} \alpha_k X^k \cdot CX^k < \epsilon.\label{bet1}
\end{align}
On the other hand, let $E_i$ be such that it supports a limit point of $\left\{ \bar{X}^K \right\}$. Then
\begin{align*}
\liminf\limits_{K \rightarrow \infty} \left\{ \frac{1}{A_K} \sum_{k=0}^K \alpha_k (E_i - X^k) \cdot CX^k \right\} = 0
\end{align*}
for otherwise Lemmas \ref{probability_mass_convergence_2} and \ref{probability_mass_convergence_7} imply that $\bar{X}^K(i) \rightarrow 0$. Therefore, combining with Lemma \ref{main_convergence_lemma}, we obtain
\begin{align*}
\lim_{K \rightarrow \infty} \left\{ \frac{1}{A_K} \sum_{k=0}^K \alpha_k (E_i - X^k) \cdot CX^k \right\} = 0,
\end{align*}
which implies that
\begin{align}
\forall \epsilon > 0 \mbox{ } \exists N > 0 \mbox{ } \forall K \geq N : (C\bar{X}^K)_i - \frac{1}{A_K} \sum_{k=0}^K \alpha_k X^k \cdot CX^k > - \epsilon.\label{bet2}
\end{align}
Combining \eqref{bet1} and \eqref{bet2}, we obtain
\begin{align}
\forall \epsilon > 0 \mbox{ } \exists N > 0 \mbox{ } \forall K \geq N : (C\bar{X}^K)_{\max} - (C\bar{X}^K)_i < \epsilon.\label{joinbed}
\end{align}
This proves the first part of the lemma. To prove the second part, let $X^*$ be a limit point of $\left\{ \bar{X}^K \right\}$ and observe that \eqref{joinbed} implies
\begin{align*}
\forall i \in \mathcal{C}(X^*) \mbox{ } \forall \epsilon > 0 \mbox{ } \exists N > 0 \mbox{ } \forall K \geq N : (C\bar{X}^K)_{\max} - (C\bar{X}^K)_i < \epsilon.
\end{align*}
Multiplying the previous inequality by $X^*(i), i =1, \ldots, n$ and summing over $i$, we obtain
\begin{align*}
\forall \epsilon > 0 \mbox{ } \exists N > 0 \mbox{ } \forall K \geq N : (C\bar{X}^K)_{\max} - X^* \cdot C\bar{X}^K < \epsilon.
\end{align*}
That is, $X^*$ is a well-supported approximate best response to $\bar{X}^K$ as claimed.
\end{proof}

\if(0)

We finally obtain the following ``clean separation'' phenomenon based on best responses:

\begin{corollary}
\label{best_response_separation_corollary}
There exists $\hat{\epsilon} > 0$ such that, for all $\epsilon < \hat{\epsilon}$, there exists an iteration $N$ such that for all $K \geq N$ every pure strategy that is a best response to a limit point of $\left\{ \bar{X}^K \right\}$ is an $\epsilon$-approximate best response to $\bar{X}^K$ and the remaining pure strategies are not.
\end{corollary}

\begin{proof}
Let
\begin{align*}
\mathcal{K}^+ = \left\{ i \in \mathcal{K}(C) \bigg| \limsup\limits_{K \rightarrow \infty} \left\{ \frac{1}{A_K} \sum_{k=0}^K \alpha_k (E_i - X^k) \cdot CX^k \right\} = 0  \right\}
\end{align*}
and
\begin{align*}
\mathcal{K}^- = \left\{ j \in \mathcal{K}(C) \bigg| \limsup\limits_{K \rightarrow \infty} \left\{ \frac{1}{A_K} \sum_{k=0}^K \alpha_k (E_j - X^k) \cdot CX^k \right\} < 0  \right\}.
\end{align*}
Lemma \ref{lovestar} implies that the pure strategies that are best responses to a limit point of $\left\{ \bar{X}^K \right\}$ are precisely those strategies in $\mathcal{K}^+$. Furthermore, Theorem \ref{persistent_best_responses} implies that for all $\epsilon > 0$ there exists an iteration $N$ such that for all $K \geq N$, $i \in \mathcal{K}^+$ implies that $E_i$ is an $\epsilon$-approximate best response to $\bar{X}^K$. Moreover, Lemma \ref{love1} implies that if $i \in \mathcal{K}^-$, there exists $\hat{\epsilon}_i > 0$ and $N_i > 0$ such that, $\forall K \geq N_i$, $E_i$ is not a $\hat{\epsilon}_i$-approximate best response to $\bar{X}^K$. Together these observations imply the corollary.
\end{proof}

The previous results do not eliminate the possibility that there exists a pure strategy that is a best response to a limit point of $\bar{X}^K$ only in the limit as $K \rightarrow \infty$. Demonstrating the existence or eliminating the possibility of such a phenomenon is an interesting question for future work.

\fi

\section*{Acknowledgments}

I would like to thank Robert Vanderbei for a helpful discussion and for reading an earlier version of this manuscript. I would also like to thank YouTube for being a true companion in my research.

\bibliographystyle{abbrvnat}
\bibliography{real}

\end{document}